\documentclass[%
 prx,amsmath,amssymb,longbibliography,
 aps,tightenlines,epsf,showpacs,superscriptaddress,twocolumn
]{revtex4-2}

\usepackage{graphicx}
\usepackage{dcolumn}
\usepackage{bm}
\usepackage[dvipsnames]{xcolor}
\usepackage{physics}
\usepackage{orcidlink}
\usepackage{mathrsfs}
\usepackage{hyperref}
\hypersetup{
    colorlinks = true,
    citecolor = magenta,
    linkcolor = magenta
}
\usepackage{algorithm}
\usepackage{algpseudocode}
\usepackage{amsthm}
\usepackage{braket}
\usepackage[english]{babel}
\newtheorem{theorem}{Theorem}[section]

\newtheorem{proposition}{Proposition}[section]
\newtheorem{lemma}{Lemma}[section]
\newtheorem{corollary}{Corollary}[section]
\newtheorem{definition}{Definition}[section]

\newcommand{\YZ}[1]{{\color{blue}[YZ: #1]}}
\newcommand{\fz}[1]{{\color{orange}[FZ: #1]}}
\newcommand{\vtwo}[1]{{\color{black}#1}}

\begin{document}

\preprint{APS/123-QED}

\title{Classical Simulability of Quantum Circuits with Shallow Magic Depth}

\author{Yifan Zhang\orcidlink{0000-0003-2448-2035}}
\email{yz4281@princeton.edu}
\affiliation{Department of Electrical and Computer Engineering, Princeton University, Princeton, NJ 08544}
 
\author{Yuxuan Zhang\orcidlink{0000-0001-5477-8924}}%
\email{quantum.zhang@utoronto.ca}
\affiliation{Department of Physics and Centre for Quantum Information and Quantum Control, University of Toronto, 60 Saint George Street, Toronto, Ontario, Canada M5S 1A7}
\affiliation{Vector Institute, W1140-108 College Street, Schwartz Reisman Innovation Campus
Toronto, Ontario, Canada, M5G 0C6}
\date{\today}

\begin{abstract}
\vtwo{Quantum magic is a necessary resource for for quantum computers to be not efficiently simulable by classical computers.}
Previous results have linked the \emph{amount} of quantum magic, characterized by the number of $T$ gates or stabilizer rank, to classical simulability. However, the effect of the \emph{distribution} of quantum magic on the hardness of simulating a quantum circuit remains open. In this work, we investigate the classical simulability of quantum circuits with alternating Clifford and $T$ layers across three tasks: amplitude estimation, sampling, and evaluating Pauli observables. In the case where all $T$ gates are distributed in a single layer, performing amplitude estimation and sampling to multiplicative error are already classically intractable under reasonable assumptions, but Pauli observables are easy to evaluate. Surprisingly, with the addition of just one $T$ gate layer or merely replacing all $T$ gates with $T^{\frac{1}{2}}$, the Pauli evaluation task reveals a sharp complexity transition from P to GapP-complete. 
Nevertheless, when the precision requirement is relaxed to 1/\text{poly}(n) additive error, we are able to give a polynomial time classical algorithm to compute amplitudes, Pauli observable, and sampling from $\log(n)$ sized marginal distribution for any magic-depth-one circuit that is decomposable into a product of diagonal gates. \vtwo{This rules out certain forms of quantum advantages in these circuits.}
Our research provides new techniques to simulate highly magical circuits while shedding light on their complexity and their significant dependence on the magic depth.
%
\end{abstract}

\maketitle

\begin{table*}
\caption{\label{tab:sumary}
Summary of the complexity of classically simulating circuits with shallow magic depth. Estimating amplitudes and Pauli observable are all up to multiplicative error in $T$-depth-one, $T$-depth-two, and $T^{\frac{1}{2}}$-depth-one circuits. Simulating diagonal magic depth one
are all up to up to $\epsilon = 1/\text{poly}(n)$ additive error and with probability $1-\delta$. Green and red items represent positive and negative results obtained in this manuscript.}
\begin{ruledtabular}
\begin{tabular}{p{0.25\linewidth} | p{0.2\linewidth} p{0.2\linewidth} p{0.35\linewidth}}
  & Amplitude& Pauli &Sampling\\
  \hline
  $T$ depth one& \textcolor{BrickRed}{GapP-complete}&  \textcolor{OliveGreen}{$O(n^3)$} &\textcolor{BrickRed}{classically hard unless $\Delta_3$P=PH}\\ 
  $T$ depth two& \textcolor{BrickRed}{GapP-complete}& \textcolor{BrickRed}{GapP-complete} &\textcolor{BrickRed}{classically hard unless $\Delta_3$P=PH}\\
 $T^{\frac{1}{2}}$ depth one& \textcolor{BrickRed}{GapP-complete}& \textcolor{BrickRed}{GapP-complete} &\textcolor{BrickRed}{classically hard unless $\Delta_3$P=PH}\\
 Diagonal magic depth one& \textcolor{OliveGreen}{$O(n^3 + \frac{n \log(2/\delta )}{ \epsilon^2})$} & \textcolor{OliveGreen}{$O(n^3 + \frac{ n \log(2/\delta )}{ \epsilon^2})$} & \textcolor{OliveGreen}{$\text{poly}(n,\delta,\epsilon)$ for $\log(n)$ marginals}\\
 \end{tabular}
\end{ruledtabular}
\end{table*}

\section{Introduction}
Classical probabilistic algorithms cannot efficiently simulate universal quantum computers -- this is a common belief underscored by many renowned examples: sampling hard distributions~\cite{aaronson2011computational,aaronson2013bosonsampling,childs2013universal,aaronson2016complexity}, solving computational problems~\cite{deutsch1992rapid,shor1994algorithms,simon1997power,shor1999polynomial}, and simulating quantum dynamics~\cite{feynman2018simulating,childs2018toward,daley2022practical}. However, certain quantum information processing tasks do not require computational universality. For example, randomized benchmarking~\cite{knill2008randomized} and certain types of quantum error correction codes~\cite{gottesman1997stabilizer}, or quantum states with topological orders~\cite{kitaev1997quantum} can be efficiently simulated classically for thousands of qubits, thanks to the Gottesman-Knill theorem~\cite{gottesman1998heisenberg}. The theorem states that the Clifford group generated by the gate set $\{H, S, CNOT\}$, despite their ability to generate substantial entanglement, can be simulated classically in polynomial time in $n$~\cite{aaronson2004improved}. As such, Clifford operations are generally considered inexpensive for classical simulation. In contrast, non-Clifford features, often referred to as ``magic", are crucial  and sometimes regarded as a scarce resource for realizing the full potential of quantum computation. Understanding the relationship between the classical hardness of simulation and the amount of magic in a quantum system is therefore essential for both theoretical insights and practical advancements in quantum computing. 

But how should we quantify magic? The most straightforward way to quantify it is by the \emph{number} of non-Clifford gates, such as the $T$ gates, in a circuit. The early seminal algorithm by Aaronson and Gottesman has a runtime that scales exponentially with the number of non-Clifford gates~\cite{aaronson2004improved}. 
Using the low-stabilizer rank approximation~\cite{garcia2012efficient,garcia2014geometry}, recent simulation algorithms drastically reduce the simulation cost when the number of $T$ gates is small~\cite{bravyi2016improved,bravyi2016trading,bravyi2019simulation,kocia2020improved,qassim2021improved}. However, these algorithms still cannot avoid the exponential runtime in the presence of an extensive amount of magic. 

Is this scaling fundamental? In this work, we partially circumvent the exponential barrier by proposing a third angle: the classical simulation cost depends on the magic \emph{depth}. Specifically, if all magic gates concentrate on one layer of the circuit and are not causally dependent on one another, then certain classical simulation tasks have only polynomial runtime even in the presence of $O(n)$ magic gates. We motivate this new angle below and explain why the shallow magic depth could be favorable for classical simulations.

\subsection{Magic as Interference in Pauli Basis}
We begin by offering insight into why magic depth should play a crucial role in classical simulability. One way to understand magic is to think of it as ``interferometers'' that generate superposition in the Pauli basis. As an example, under the evolution of a $T$ gate, the Pauli $X$ and $Y$ operators become superimposed.
\begin{align}
    T X T^\dag &= \frac{1}{\sqrt{2}}(X+Y) \\
    T Y T^\dag &= \frac{1}{\sqrt{2}}(-X+Y)
\end{align}
This can be compared with the Hadamard gate which generates superposition in the computational basis: $H\ket{0}=\frac{1}{\sqrt{2}}(\ket{0}+\ket{1})$ and $H\ket{1}=\frac{1}{\sqrt{2}}(\ket{0}-\ket{1})$. Now, suppose there is only one layer of Hadamard gates in the circuit, sandwiched by other gates that do not generate superposition (e.g. phase gates and permutation gates; also called ``almost classical gates'', see Definition \ref{def:almost_classical}), then classically simulating this circuit is trivial, as the final state is a uniform superposition of all computational basis, each carrying a phase that can be efficiently computed. In other words, no interference can happen with only one layer of Hadamard gates.

On the other hand, two layers of Hadamard gates can generate interference and render the final state classically intractable. In fact, this notion of generating interference using multiple layers of Hadamard gates is already investigated and characterized by the Fourier Hierarchy $\mathcal{FH}_m$ \cite{shi2005quantum}. Informally, $\mathcal{FH}_m$ are problems that can be solved using $m$ layers of Hadamard gates. $\mathcal{FH}_1= \text{BPP}$ because the output probability is uniform. Notably, $\mathcal{FH}_2$ already contains hard problems such as factoring \cite{kitaev1995quantum}, demostrating the power of quantum computing with only two layers of Hadamard gates.

The drastic difference between $\mathcal{FH}_1$ and $\mathcal{FH}_2$ motivates us to ask a similar question in the context of magic: if all magic gates concentrate at one layer in a circuit, could the classical simulation become simplified because of the lack of interference? On the other hand, if there are two layers of magic gates, could the classical simulation suddenly become hard due to interference?

It turns out that the Pauli basis behaves differently from the computational basis. The reason is that a pure initial state $\ket{0^n}$ is sparse in the computational basis but has exponential support in the Pauli basis. $\ketbra{0^n}{0^n}$ is a stabilizer state generated by all local $Z$ operators. Therefore, $\ketbra{0^n}{0^n}$ decomposes into the sum of the $2^n$ Pauli operators that contain only $Z$ or $I$. Thus, even after one layer of magic gates, the Pauli decomposition of the state already becomes complicated. As we will see later, even with one layer of magic gates, certain computational tasks already become classically intractable, while some other tasks admit polynomial-time algorithms.

\subsection{Summary of Results}   
%


We study the classical simulability of quantum circuits with one or two layers of magic gates. We analyze the computational complexity of three simulation tasks: amplitude estimation, sampling, and estimating Pauli observable. We present the magic-depth-one circuit that we consider in this work in Fig. \ref{fig:td1_setup_iqp}(a,b). The unitary dynamics $U=U_{c,r} \prod_i D_i U_{c,l}$ consists of two Clifford unitaries $U_{c,l}$ and $U_{c,r}$ sandwiching a layer of magic gates $\prod_i D_i$, each acting on $O(1)$ qubits. We note that while each $D_i$ is local, $U_{c,l}$ and $U_{c,r}$ can be arbitrarily non-local. Fig. \ref{fig:td1_setup_iqp}(a) shows the task of computing amplitudes $\bra{0} U \ket{0}$, while Fig. \ref{fig:td1_setup_iqp}(b) shows the task of computing Pauli observable $\bra{0} U^\dag P U \ket{0}$. Notice that we can always remove $U_{c,r}$ by replacing $P$ with $U_{c,r}^\dag P U_{c,r}$ which is also a Pauli operator.
We also consider estimating amplitudes and Pauli observable up to different precision. Depending on the number of magic layers, the simulation tasks, and the precision requirement, the complexity is drastically different. 

\paragraph{\vtwo{Hardness of classical simulations at multiplicative error.}} To begin with, we prove that computing amplitudes and sampling to multiplicative error are already classically hard even at $T$-depth-one. This is accomplished by a newly devised ``parallelization trick'' that reduces a degree-three ``instantaneous quantum polynomial'' (IQP) circuit into a $T$-depth-one circuit, and utilizes the known hardness result for sampling complexity for the IQP circuit. On the contrary, we give a polynomial circuit for exactly computing the Pauli observable for circuits of $T$ depth one, making use of the symmetry the $T$ gate possesses, as it belongs to the third level of Clifford Hierarchy~\cite{gottesman1999demonstrating}. Surprisingly, by adding one layer of $T$ gate to the circuit, or simply substituting all $T$ gates with $T^{1/2}$, the hardness of Pauli evaluation to multiplicative accuracy goes through a sharp transition, from P to GapP-complete.

\vtwo{\paragraph{An efficient classical algorithm at additive error in magic-depth-one circuits.} In addition, when one demands $1/\text{poly}{(n)}$ additive error instead of a multiplicative error, then estimating amplitudes and Pauli observable as well as sampling from a $\log(n)$ sized marginal distribution become classically easy at magic-depth-one for arbitrary diagonal magic gates. We show it by constructing an explicit classical algorithm, shown in Algorithm \ref{alg:additive}, to perform these tasks. This algorithm is practical and has the run-time roughly equivalent to sampling stabilizer states. This also rules out the possibility of quantum advantages in diagonal magic-depth-one circuits without taking advantage of a marginal distribution with $\omega(\log(n))$ qubits. The above main results are summarized in Tab.~\ref{tab:sumary}.}

\paragraph{\vtwo{A simulation algorithm for circuits with shallow magic depth.}} Lastly, we provide a path-integral algorithm for circuits with more than one layer of magic gates. While the algorithm scales exponentially in the system size, the scaling in the number of magic layers is sub-exponential, \vtwo{rendering it favorable over exists techniques of low-stabilizer-rank decompositions in circuits with extensive magic but shallow magic depth}.




\begin{figure*}
\includegraphics[width=\linewidth]{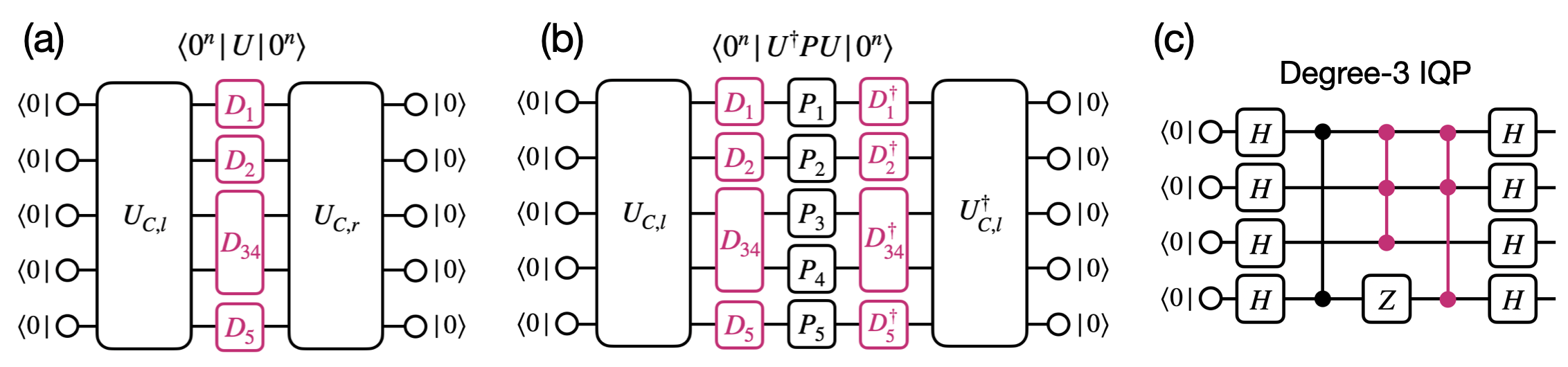}
\caption{\label{fig:td1_setup_iqp} 
 \textbf{Circuits considered in this paper.} Clifford gates are marked black and non-Clifford gates are marked in magenta color. (a) The task of computing amplitude in magic-depth-one circuits. $D_i$ are magic gates acting on $O(1)$ qubits, while $U_{c,l}$ and $U_{c,r}$ are Clifford unitaries sandwiching the magic layer. (b) The task of computing Pauli observable in magic-depth-one circuits. Note that we remove $U_{c,r}$ and replace $P$ with $U_{c,r}^\dag P U_{c,r}$. (c) An example of the degree-three IQP circuits. The magic gates $CCZ$ are in magenta color.}
\end{figure*}

\section{Hardness of Computing Amplitudes in Magic-Depth-One Circuits}
In this section, we establish the hardness of computing amplitudes in magic-depth-one circuits. There are many approaches to establish such hardness, and we will show the hardness by connecting to the IQP circuit, a candidate for quantum advantage demonstrations~\cite{shepherd2009temporally,bremner2011classical,bremner2016average,bluvstein2024logical} where the hardness of computing amplitude is well known. An IQP circuit can be written in the format: $H^{\otimes n}D_{IQP}H^{\otimes n}$, where $H$ represents the Hadamard gate and $D_{IQP}$ is a generic diagonal gate. For our purpose, it suffices to consider a subset of IQP circuits, the so-called degree-three IQP:

\begin{definition}
A degree-three IQP circuit has its $D_{IQP3}$ synthesized from only $Z$, $CZ$, and $CCZ$ gates.
\end{definition}
We show an example of the degree-three IQP circuit in Fig. \ref{fig:td1_setup_iqp}(c). The name comes from the fact that the phase $f(x)=\pm 1$ that $D_{IQP3}$ applies to a basis state $\ket{x}$ can be computed from a third-degree polynomial over the finite field $\mathbb{F}_2$. The $Z$, $CZ$, and $CCZ$ gates correspond to the first, second, and third degree terms in the polynomial \cite{maslov2024fast}. It is known that computing the amplitude of degree-three IQP circuits, even up to a small multiplicative error, is GapP-complete~\cite{bremner2016average,fenner1994gap,ehrenfeucht1990computational}.
We now prove the hardness of computing amplitudes of $T$-depth-one circuits by providing an algorithm that compiles any degree-three IQP circuit into a $T$-depth-one circuit.

\begin{proposition} \label{iqp_sharp_p}
    Computing $\bra{0}H^{\otimes  n}D_{IQP3}H^{\otimes  n}\ket{0}$ up to a $1$ multiplicative error is GapP-complete \cite{bremner2016average}.
\end{proposition}

The complexity class GapP is defined as follows: given a nondeterministic polynomial-time Turing machine $M$, let $acc_M(x)$  be the number of accepting paths of $M$ on input $x$, and $ rej_M(x) $ be the number of rejecting paths. GapP is the class of functions $f(x)$ such that
\begin{equation}
    f(x) = acc_M(x) - rej_M(x)
\end{equation}
In the IQP setup, $x$ is the classical description of the IQP circuit and $f(x)$ is the amplitude. GapP is closely related to the counting class \#P. 
\vtwo{
The $1$ multiplicative error means that the estimate $\tilde{z}$ of $z$ deviates by at most $|\tilde{z} - z| \le z$. This means that when $z$ is small, the absolute error is small accordingly. Note that the multiplicative error of $1$ here differs from the multiplicative error of $\frac{1}{2}$ in \cite{bremner2016average} (see Proposition 8 therein). This is because Ref. \cite{bremner2016average} considers the task of finding the absolute value of the amplitude. On the other hand, we consider the signed amplitude here, and it is known that computing the sign alone is already sufficient to determine the amplitude exactly. This is done through a binary search, discussed in \cite{feffermanPCMI}. Therefore, having a multiplicative error of $1$ is sufficient as we can already extract the sign faithfully.
}

We will now establish the hardness of computing the amplitude in $T$-depth-one circuits. The proof is based on compiling any degree-three IQP circuit to one layer of $T$ gates. Because of Proposition \ref{iqp_sharp_p}, it follows that computing the amplitude of $T$-depth-one circuits, even up to a multiplicative error, is GapP-hard.

\subsection{Parallelization Trick}
We now give a procedure to compile any degree-three IQP circuit to one layer of gates of the form: $T^k$, where $k$ is some integer. As an initial step, we use a parallelization trick, shown in Fig. \ref{fig:parallelization}, to put all the diagonal gates in $D_{IQP3}$ in one layer. The parallelization trick works as follows. For each diagonal gate supported on a set of qubits, we introduce an equal number of ancilla initialized to $\ket{0}$, and then apply the $CNOT$ gates controlled by the original data qubits and target at the ancilla. For example, to parallelize $D_{23}$, we introduce two ancilla (the bottom two blue qubits) and then apply two $CNOT$ gates. We repeat the above steps for all diagonal gates. After that, we apply all diagonal gates to the corresponding ancilla simultaneously. Finally, we repeat the $CNOT$ gates to clean the ancilla, which means that the ancilla returns to $\ket{0}$ regardless of the state of the data qubits. We show that the new circuit after parallelization has the same effect as the original circuit.

\begin{figure}
\includegraphics[width=\linewidth]{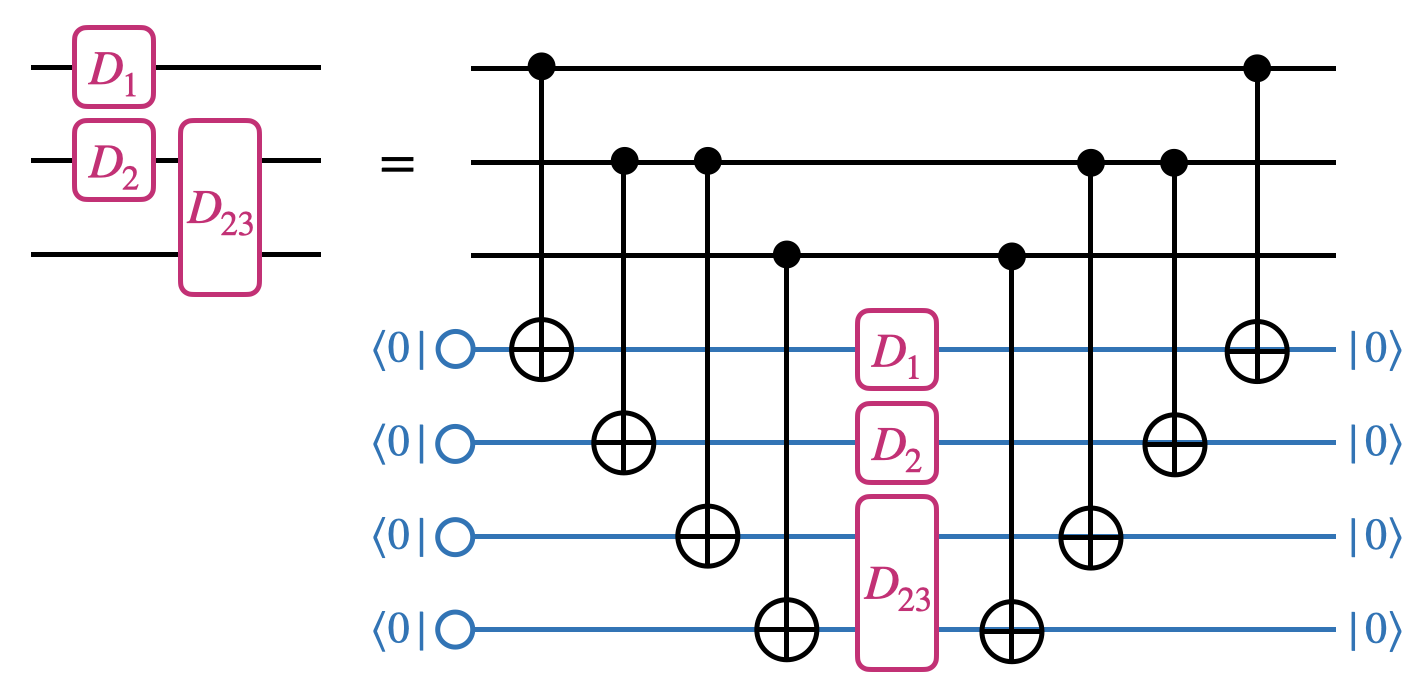}
\caption{\label{fig:parallelization}\textbf{The parallelization trick.} One introduces a set of ancilla (shown in blue) for each diagonal gate $D_i$ and copy the bitstring values to the ancilla. The phase gates are then applied simultaneously. Ancilla are cleaned in the end.}
\end{figure}

\begin{lemma}
    The parallelization trick in Fig. \ref{fig:parallelization} is equivalent the original circuit when $D_i$ are diagonal gates.
\end{lemma}
\begin{proof}
 We show that the matrix equation in Fig. \ref{fig:parallelization} is correct term by term in the computational basis. Suppose we input a state $\ket{x}$, where $x$ is a bitstring. We should expect the output to be multiplied by all phases of each diagonal gate $D_i$.
 \begin{equation}
     \prod_i D_i \ket{x} = (\prod_i  e^{i \phi_{D_i}(x)})  \ket{x}
 \end{equation}
Where $e^{i \phi_{D_i}(x)}=\bra{x} D_i \ket{x}$ denotes the phase $D_i$ applies to $\ket{x}$. Next, we consider the circuit on the right-hand side. Suppose we want to parallelize $I$ diagonal gates. We initialize a set of ancilla $\ket{0}_{A_i}$, $i=1,2, \ldots, I$. After the initial layers of $CNOT$ gates, we have
\begin{equation}
    \ket{x}_D \otimes \ket{x_1}_{A_1} \otimes \ket{x_2}_{A_2} \otimes \ldots \otimes \ket{x_I}_{A_I}
\end{equation}
Where $\ket{x}_D$ denotes the original data qubit and $\ket{x_i}_{A_i}$ denotes the ``copy" of a subset of the bitstring $x$ supporting $D_i$. For instance, in Fig. \ref{fig:parallelization}, the top ancilla (in blue) copies the first bit of $x$, the second ancilla copies the second bit of $x$, and the last two ancilla copy the second and the third bit of $x$. After copying the data qubit, we apply the diagonal gates to get
\begin{equation}
  (\prod_i  e^{i \phi_{D_i}(x)})  \ket{x}_D \otimes \ket{x_1}_{A_1} \otimes \ket{x_2}_{A_2} \otimes \ldots \otimes \ket{x_I}_{A_I}
\end{equation}
Here the phase is identical to the original circuit because $D_i \ket{x_i}_{A_i} = D_i \ket{x}_D$. Finally, the final layer of $CNOT$ gates resets all the ancilla to $\ket{0}_{A_i}$ without affecting the phase. Therefore, the circuit results in $(\prod_i  e^{i \phi_{D_i}(x)}) \ket{x}_D \ket{0}_{A_1 \ldots A_I}$ which is identical to the original circuit.
\end{proof}

Using the parallelization trick, we can put all the diagonal gates in $D_{IQP3}$ in one layer. Since the $Z$ and $CZ$ gates are Clifford, the hardness of the simulation comes from the presence of $CCZ$ gates. In the next subsection, we will show how to compile $CCZ$ gates into one layer of $T$ gates.

\subsection{Parallelizing Almost Classical Gates}
In this section, we discuss how to compile $CCZ$ gates into one layer of $T$ gates. We borrow the technique from \cite{selinger2013quantum}. First, the $CCZ$ gate can be synthesized from $CNOT$ and $T$ gates, shown in Fig. \ref{fig:ccz}(a). We observe that both $CNOT$ and $T$ gates are ``almost classical gates'' 
\begin{definition}\label{def:almost_classical}
    An almost classical gate maps any computational basis vector to some other computational basis vector with a phase. In other words, an almost classical gate $U$ can be written in the following form.
    \begin{equation}
        U = \sum_x  e^{i \phi(x)}\ketbra{f(x)}{x}
    \end{equation}
    Where $f(x): \{0,1\}^n \rightarrow \{0,1\}^n$ denotes a bijection of bitstrings and $e^{i \phi(x)}: \{0,1\}^n \rightarrow U(1)$ denotes the phase corresponding to each bitstring.
\end{definition}
$CNOT$ is a permutation in the computational basis and $T$ is a diagonal phase gate, so they are both almost classical gates. An observation is that the composition of almost classical gates is also an almost classical gate.
\begin{proposition}
The product of any two almost classical gates is also an almost classical gate
\end{proposition}
We now state the lemma that allows us to put $T$ gates in one layer. We will state this in a more generic form where we wish to put some generic diagonal gates $D$ into one layer.
\begin{lemma}\label{almost_classical_one_layer}
    Given a gate set consists of almost classical gates, including CNOT. Suppose a diagonal gate $D$ is part of the gate set. Denote a unitary $U=U_IU_{I-1} \ldots U_{2}U_{1}$, where each gate $U_i$ is chosen from the gate set. Then $U$ can be compiled, after including the ancilla, so that all the gates of the form $D^k$, where $k$ is some integer, are in one layer.
\end{lemma}
\begin{proof}
    The proof is essentially a generalization of the proof of Theorem 4.1 in \cite{selinger2013quantum}. Through induction, we decompose $U$ into $U=A_1A_2$, where $A_1$ is diagonal and with at most one layer of $D^k$ gates (the components of $A_1$ need not be diagonal; only the overall product needs to be diagonal). $A_2$ contains no $D^k$ gate. 

We now construct $A_1$ and $A_2$ by induction. We initialize $A_1$ and $A_2$ to the identities. In the $i$-th step, we have $A_1'$ and $A_2'$ from the previous step and apply $U_i$ to get $U_i A_1' A_2'$. Depending on the type of $U_i$ we perform the following actions.
\begin{enumerate}
    \item If $U_i$ is not a $D^k$ gate, then let $A_1 = U_i A_1' U_i^\dag$, $A_2=U_i A_2'$. Since $U_i$ is almost classical, if $A_1'$ is diagonal, then $A_1$ is also diagonal.
    \item If $U_i$ is not a $D^k$ gate, without loss of generality we assume $D^k$ applies to a single qubit $i$, then let $A_2=A_2'$ and let $A_1$ be
    \begin{equation}
        \scalebox{0.3}{\includegraphics{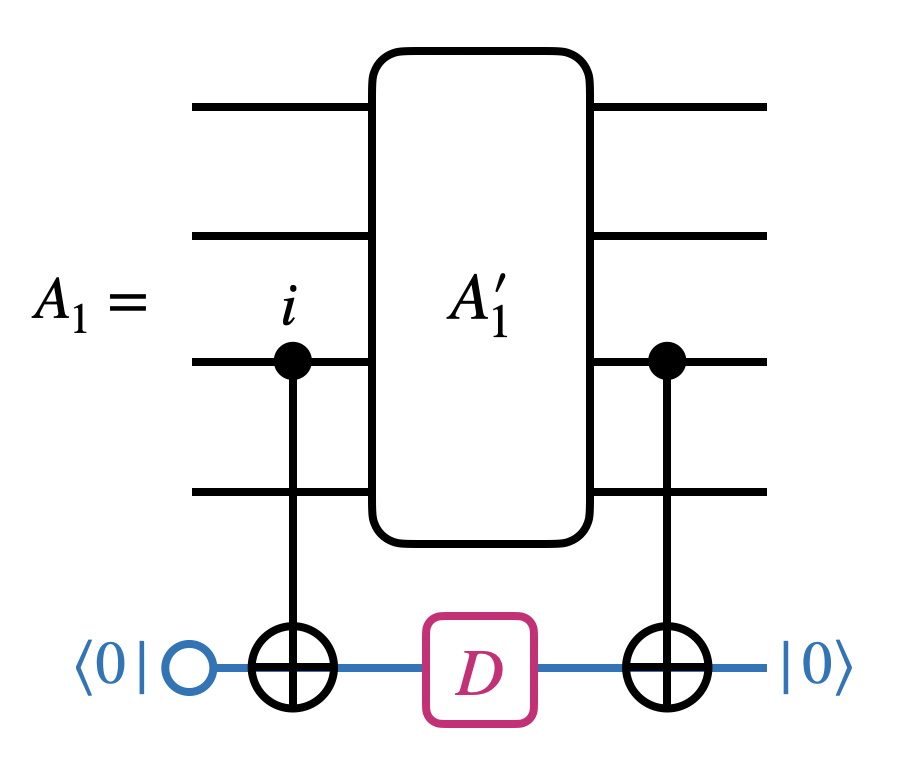}}
    \end{equation}
    Since $A_1'$ is diagonal, the above circuit is equivalent to applying $U_i A_1'$ and $A_1$ is also diagonal. Since $A_1'$ has $D$ depth one, $A_1$ also has $D$ depth one.
\end{enumerate}

    At the end of the induction, $A_1A_2$ has only one layer of $D^k$ gates in $A_1$.
\end{proof}

\subsection{Proof of Hardness}
We are now ready to show the hardness of computing amplitudes of $T$-depth-one circuits. We first show that any degree-three IQP circuit can be compiled to one layer of $T$ gates, after appending ancilla proportional to the number of diagonal gates.
\begin{lemma}
    Any degree-three IQP circuit with $d$ layers of diagonal gates can be compiled into one layer of T gates, after appending $O(nd)$ pure ancilla initialized to $\ket{0}$.
\end{lemma}
\begin{proof}\label{d3_iqp_td1}
    We parallelize all the diagonal gates using the trick shown in Fig.  \ref{fig:parallelization}. Next, the $CCZ$ gate can be decomposed into $CNOT$ gates and $T$ gates as shown in Fig. \ref{fig:ccz}(a). Then, using Lemma \ref{almost_classical_one_layer}, one can compile individual $CCZ$ gates to have $T$ depth one. For concreteness, we show the compilation of $CCZ$ with one layer of $T$ gates in Fig. \ref{fig:ccz}(b).
\end{proof}
After compiling the degree-three IQP circuit into one layer of $T$ gates, we can establish our first hardness result.
\begin{theorem}\label{thm:hardness_amplitude}
    Given a circuit with one layer of T gates $U=U_{c,r} (\prod_i T_i^{k_i}) U_{c,l}$, where $T_i$ acts on the qubit $i$ and $k_i$ denotes some integer power, then computing $\Re[\bra{0}U\ket{0}]$ up to a multiplicative error of $1$ is GapP-complete.
\end{theorem}
\begin{proof}
This problem is in GapP because stabilizer states $U_{c,l} \ket{0}$ and $\bra{0} U_{c,r}$ can be written in the computational basis in $O(n^3)$ time \cite{dehaene2003clifford}. With such representations, one can compute $\Re[\bra{0}U_{c,r} (\prod_i T_i^{p_i}) U_{c,l}\ket{0}]$ by summing up all the real contributions from each basis vector which is in GapP.

To show the GapP-hardness, first use Lemma \ref{d3_iqp_td1} to compile any degree-three IQP circuit $H^{\otimes n}D_{IQP3}H^{\otimes n}$ to the T-depth-one circuit $U$. After the compilation, we have introduced some ancilla which are initialized in $\ket{0}_A$ and always return to $\ket{0}_A$ after the computation. Therefore, computing $\bra{0}_DH^{\otimes n}D_{IQP3}H^{\otimes n}\ket{0}_D$ is equivalent to computing $\bra{0}_D\bra{0}_A U\ket{0}_A\ket{0}_D$. Then, Proposition \ref{iqp_sharp_p} immediately implies the GapP-completeness. 
\end{proof}

Since a quantum computer is a sampling machine, one should really quantify the hardness of classically sampling the distribution. Since the hardness of sampling from degree-three IQP circuits is known under some plausible complexity conjectures, it follows that sampling $T$-depth-one circuits up to a small statistical error is also classically hard. We quantify this error using the total variation distance, defined as follows.
\begin{equation}\label{tvd}
    \delta(p(x),q(x)) = \frac{1}{2} \sum_x |p(x)-q(x)|
\end{equation}
Where $p(x)$ and $q(x )$ are the two probability distributions. \vtwo{The hardness of sampling from degree-three IQP circuits is quoted below.
\begin{proposition}\label{prop:iqp_samp_hardness}
    (Theorem 1 of \cite{bremner2016average}) If there exists a probabilistic classical polynomial-time algorithm to sample from the distribution of any degree-three IQP circuit up to a total variation distance of $\frac{1}{384}$, then under the complexity-theoretic Conjecture 3 in \cite{bremner2016average}, the polynomial hierarchy collapses to the third level.
\end{proposition}
}
We use the hardness of sampling degree-three IQP circuits to show that sampling from $T$-depth-one circuit classically would imply the collapse of the polynomial hierarchy.

\begin{theorem}
    Given a circuit with one layer of T gates $U=U_2 (\prod_i T_i^{k_i}) U_1$, where $T_i$ acts on the qubit $i$ and $k_i$ denotes some integer power, then under the Conjecture 3 from \cite{bremner2016average}, if there exists a classical algorithm to sample from $U\ket{0}$ in the computational basis up to a total variation distance of $\frac{1}{384}$, then the polynomial hierarchy collapses to the third level. \end{theorem}
\begin{proof}
We again use Lemma \ref{d3_iqp_td1} to compile any degree-three IQP circuit to the T-depth-one circuit $U$. We know that the ancilla are initialized in $\ket{0}_A$ and always return to $\ket{0}_A$ after the computation. Therefore, sampling from $U \ket{0}_D \ket{0}_A$ results in the distribution of all zeros on the ancilla qubits and the IQP distribution on the data qubits. Thus, the hardness of sampling from $T$-depth-one circuits follows from Proposition \ref{prop:iqp_samp_hardness}.
\end{proof}
\begin{figure}
\includegraphics[width=\linewidth]{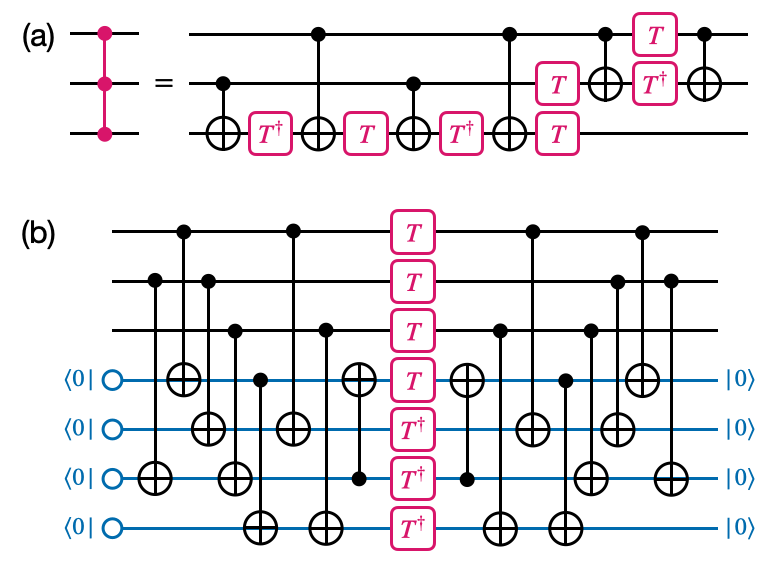}
\caption{\label{fig:ccz}\textbf{Decomposition of the $CCZ$ gate.} (a) Decomposing $CCZ$ gate into $CNOT$ gates and $T^{\pm 1}$ gates. (b) Compiling the circuit in (a) into one layer of $T^{\pm 1}$ gates. }
\end{figure}

In this section, we have shown that even if we restrict the magic gates to be $T$ gates, computing amplitude still remains hard. Such restriction, however, has non-trivial consequences for other computational tasks. We will see later on that restricting magic gates to $T$ gates renders the computation of Pauli observable classically efficient due to the special property of $T$ gates, whereas for more generic magic gates, computing Pauli observable still remains classically hard.
\section{Computing Pauli Observable in Shallow Magic Depth Circuits}
\subsection{Exact Computation of Pauli Observable in $T$ Depth 1}
After seeing the hardness of computing the amplitude of $T$-depth-one circuits, we switch our task to computing Pauli observable (Fig. \ref{fig:td1_setup_iqp}(b)). 
Surprisingly, computing Pauli observable in $T$-depth-one circuits is classically easy. This is because $T$ gate belongs to the third level of the Clifford hierarchy and possesses some special symmetry.
\begin{definition}
 We define the first level of the Clifford Hierarchy $\mathcal{CH}_1$ as the Pauli group. The $l$-th level of the Clifford Hierarchy $\mathcal{CH}_l$ is defined as a collection of gates $U$ satisfying the following property: for any Pauli operator $P$, $U P U^\dag$ is in the  $(l-1)$-th level of the Clifford Hierarchy $\mathcal{CH}_{l-1}$.
\end{definition}
Following the above definition, the second level $\mathcal{CH}_2$ is the Clifford group, and the third level of the Clifford Hierarchy $\mathcal{CH}_{3}$ contains gates $U$ such that $U P U^\dag$ is in the Clifford group. Both the $T$ gate and the $CCZ$ gate are in $\mathcal{CH}_{3}$. Notably, unlike the Pauli group or the Clifford group, $\mathcal{CH}_{3}$ does not form a group and gives a sufficient gate set for universal quantum computation. The current magic state injection protocols also teleport gates from $\mathcal{CH}_{3}$ using only Clifford operations \cite{gottesman1999demonstrating,knill2004fault,bravyi2005universal}. This is possible exactly because $U P U^\dag$ is in the Clifford group.

The above discussion shows that $\mathcal{CH}_{3}$ is powerful enough yet possesses special structures to exploit. In the context of magic-depth-one circuits, we show that when the magic gates are from $\mathcal{CH}_{3}$, then computing Pauli observable becomes classically efficient.

\begin{theorem}\label{thm:3ch_pauli}
    Given a circuit with one layer of non-Clifford gate $U= \prod_i (\tilde{U}_i) U_{c,l}$, where each $\tilde{U}_i$ is in the third level of the Clifford hierarchy $\mathcal{CH}_{3}$, then there exists a classical algorithm to compute any Pauli observable $\braket{P}=\bra{0} U^\dag P U \ket{0}$ in $O(n^3)$ time.
\end{theorem}

\begin{proof}
We prove by giving the classical algorithm explicitly. The key observation, visualized in Fig. \ref{fig:3ch_4ch}(a), is that after evolution of $\prod_i (\tilde{U}_i)$, $P$ becomes a product of local Clifford operators with a particular phase to ensure herminicity.
\begin{align}
 \prod_i (\tilde{U}_i^\dag)  P \prod_i (\tilde{U}_i)  &= \prod_i U_{ci} \\
 U_{ci} &= \tilde{U}_i^\dag P_{\text{supp}(i)} \tilde{U}_i
\end{align}
Where $P_{\text{supp}(i)}$ is the part of $P$ on the support of $\tilde{U}_i$, and $U_{ci}= \tilde{U}_i^\dag P_{\text{supp}(i)} \tilde{U}_i$ is the local Hermitian Clifford operator generated by $\tilde{U}_i$. After the evolution, the problem of computing Pauli observable becomes evaluating the amplitude of a Clifford unitary, shown in Fig. \ref{fig:3ch_4ch}(b), under a particular phase convention of $U_{ci}$ (the phase of $U_{c,l}$ can be chosen arbitrarily as $U_{c,l}^\dag$ cancels the phase out).

While computing the squared amplitude of a Clifford circuit is well known \cite{aaronson2004improved}, computing the amplitude and keeping track of the phase takes a bit of extra work. We will employ the technique of \cite{bravyi2016improved}. We write $U_{c,l}\ket{0}$ in the computational basis, following \cite{dehaene2003clifford}:
\begin{equation}\label{clifford_comp_basis}
    U_{c,l}\ket{0} = \sum_{x \in \mathcal{K}} e^{i\frac{\pi}{4} q(x)} \ket{x}
\end{equation}
where $\mathcal{K} \subseteq \mathbb{F}_2^n$ denotes an affine subspace and $q(x): \mathcal{K} \rightarrow \mathbb{Z}_8$ denotes a quadratic form (we follow the notation in \cite{bravyi2016improved}). One can choose an arbitrary sign convention for $U_{c,l}\ket{0}$ as it will be cancelled out later in the inner product. Next, we compute $ (\prod_i U_{ci}) U_{c,l}\ket{0}$. It can again be written in the computational basis:
\begin{equation}
    (\prod_i U_{ci})U_{c,l}\ket{0} = \sum_{x \in \mathcal{K'}} e^{i\frac{\pi}{4} q'(x)} \ket{x}
\end{equation}
Crucially, the sign convention of $ (\prod_i U_{ci}) U_{c,l}\ket{0}$, in other words, the constant term in the quadratic form $q'(x)$, is completely fixed by $U_{c,l}\ket{0}$ and $(\prod_i U_{ci})$. See \cite{garcia2014hybrid} for the action of the Clifford gates in the computational basis. With $\mathcal{K}$, $q(x)$, $\mathcal{K}'$, and $q'(x)$, one can calculate the inner product $\bra{0} U_{c,l}(\prod_i U_{ci})U_{c,l}\ket{0}$ in $O(n^3)$ time, using the algorithm in Appendix C of \cite{bravyi2016improved}.
\end{proof}

\begin{figure*}
\includegraphics[width=\linewidth]{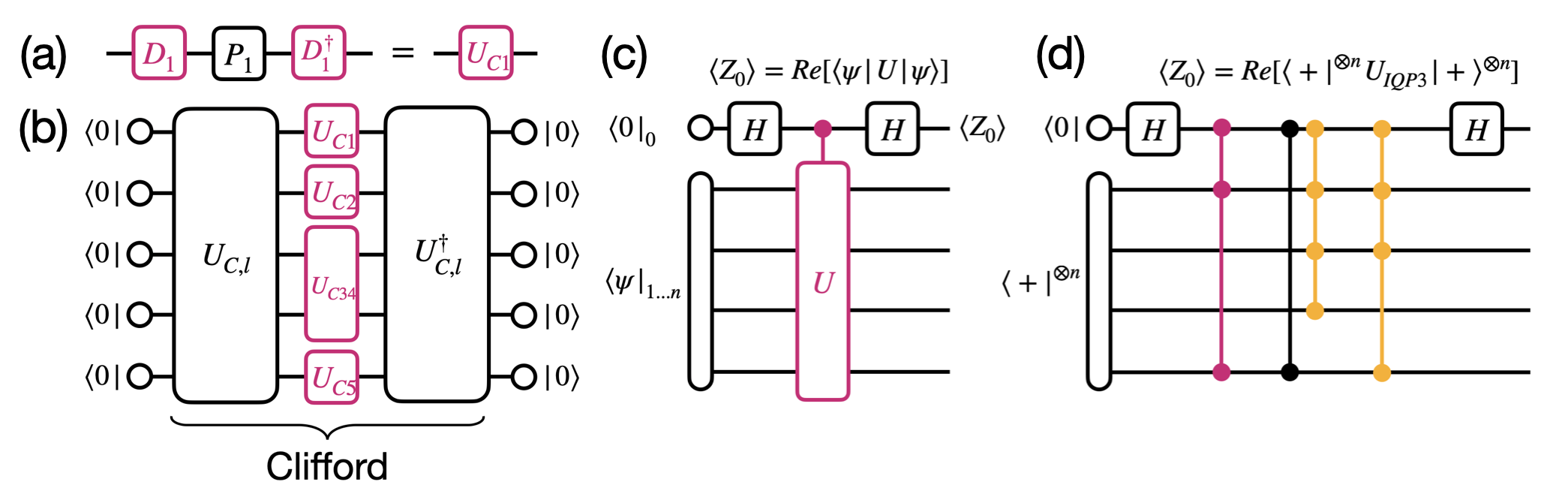}
\caption{\label{fig:3ch_4ch}\textbf{Computing Pauli observable and the Hadamard test.} (a) Evolving Pauli operators with gates from $\mathcal{CH}_3$ turn them into Hermitian Clifford operator.  (b) The Hadamard test reduces computing ampiltudes to computing Pauli observable. (c) The Hadamard test that compute the amplitude of the degree-three IQP circuit shown in Fig. \ref{fig:td1_setup_iqp}(c). $CCZ$ gates are in magenta and $CCCZ$ gates are in yellow color.}
\end{figure*}
\subsection{Hardness of Computing Pauli Observable in Magic Depth 2}
After seeing a classical polynomial-time algorithm to compute Pauli observables in $T$-depth-one circuits, one may ask how far this result can be extended. For instance, can Pauli observables in $T$-depth-two circuits still be classically computed? In addition, if one replaces the $T$ gates with more generic phase gates, does the classical algorithm still hold? We give negative answers to both questions. Specifically, we show that by having either (1)  magic gates from the fourth level of the Clifford hierarchy $\mathcal{CH}_4$ or (2) two layers of $T$ gates, computing Pauli observables becomes GapP-hard.

On a high level, we show the hardness using the Hadamard test, which reduces the task of computing amplitudes to the task of computing Pauli observables. We next show that with gates from $\mathcal{CH}_4$ or two layers of $T$ gates, one can synthesize Hadamard tests that compute the amplitude of any degree-three IQP circuits. Therefore, the hardness of computing Pauli observable follows from the hardness of computing amplitudes of degree-three IQP circuits.

\subsubsection{Hadamard Test}

We first describe the technique of the Hadamard test, which allows us to reduce computing amplitudes to computing Pauli observables. The Hadamard test is shown in Fig. \ref{fig:3ch_4ch}
(c). The circuit contains a clean ancilla $\ket{0}_0$ and an arbitrary initial state $\ket{\psi}_{1\ldots n}$. A Hadamard gate is first applied to the clean ancilla; then a controlled-$U$ gate is applied between the ancilla and the state; finally, a Hadamard gate is applied to the ancilla again. After the circuit, the state evolves to
    \begin{equation}
        \frac{1}{2} \ket{0}_0 (\ket{\psi} - U\ket{\psi})_{12...n} + \frac{1}{2}\ket{1}_0 (\ket{\psi} + U\ket{\psi})_{12...n}
    \end{equation}
One can explicitly verify that $\langle Z_0 \rangle = Re[\bra{\psi}U\ket{\psi}]$, thus evaluating Pauli observable in this circuit allows one to compute the amplitude $Re[\bra{\psi}U\ket{\psi}]$. By setting $\ket{\psi} = \ket{+}^n$ and $U=D_{IQP3}$, computing the amplitude of a degree-three IQP circuit reduces to evaluating the Pauli observable of a Hadamard test. Thus, one would expect that computing the Pauli observable of a Hadamard test is hard.

Since $D_{IQP3}$ contains $Z$, $CZ$, and $CCZ$ gates, one can explicitly construct controlled-$D_{IQP3}$ by replacing each with $CZ$, $CCZ$, and $CCCZ$ gates. Crucially, the $CCCZ$ gate is in the fourth level of the Clifford hierarchy $\mathcal{CH}_4$. This means that the previous algorithm, which relies on the property of $\mathcal{CH}_3$, does not apply to the Hadamard test of the degree-three IQP circuit. 

The remaining task is to show that the controlled-$D_{IQP3}$ can be decomposed into one layer of some magic gates or two layers of $T$ gates. Since the controlled-$D_{IQP3}$ consists of diagonal gates, they can be parallelized to one layer, where $CZ$ is Clifford and $CCZ$ can be compiled to one layer of $T$ gates. Therefore, one has to primarily concern about the decomposition of $CCCZ$ gates.

In the next subsection, we will show that
\begin{enumerate}
    \item $CCCZ$ gate can be compiled to one layer of $T^{\pm \frac{1}{2}}$ gates which is also in $\mathcal{CH}_4$. $T^{\pm \frac{1}{2}}$ is defined as
\begin{equation}
    T^{\pm \frac{1}{2}} = \ketbra{0}{0} \pm e^{i \frac{\pi}{8}}\ketbra{1}{1}
\end{equation}
One can see that $T^{\pm \frac{1}{2}}$ rotates along the z-axis with a smaller angle, and applying $T^{\pm \frac{1}{2}}$ twice gives $T^{\pm 1}$. 
    \item $CCCZ$ gate can be compiled to two layers of $T$ gates.
\end{enumerate}

To quickly see the results, we explicitly show the two decompositions in Fig. \ref{fig:cccz}. In Fig. \ref{fig:cccz}(a), we decompose a $CCCZ$ gate into two layers of $CCZ$ and $CS$ gates, separated by the $X^{-\frac{1}{2}}$ gate (shown in green). A $CS$ gate is defined as
\begin{equation}
    CS = \ketbra{00}{00} + \ketbra{01}{01} + \ketbra{10}{10} + i \ketbra{11}{11}
\end{equation}
The $CS$ gate also belongs to $\mathcal{CH}_3$. In the next section, we will show that $\mathcal{CH}_3$ can also be decomposed into $CNOT$ and $T$ gates. Thus, using Lemma \ref{almost_classical_one_layer}
, we can compile two layers of $CCZ$ and $CS$ gates separately into two layers of $T$ gates. Notice that $X^{-\frac{1}{2}}=H S^{-1} H$ is not almost classical. Here, $S=\ketbra{0}{0}+i\ketbra{1}{1}$ is the Clifford phase gate. Thus, one cannot apply Lemma \ref{almost_classical_one_layer} to both layers of $CCZ$ and $CS$ gates together. That is why we need two layers of $T$ gates to synthesize a $CCCZ$ gate.

In Fig. \ref{fig:cccz}(b), we decompose a $CCCZ$ gate into products of $CNOT$ and $T^{\pm \frac{1}{2}}$. Then, applying Lemma \ref{almost_classical_one_layer}, we compile the circuit to put $T^{\pm \frac{1}{2}}$ to one layer.

In fact, we will establish two generic results concerning (1) generating all diagonal gates in $\mathcal{CH}_l$ using one layer of small-angle rotations in $\mathcal{CH}_l$ and (2) generating multi-controlled $Z$ gate $C^{l-1}Z$ using two layers of gates from $\mathcal{CH}_m$, where $m<l$. The two decompositions of the $CCCZ$ gate follow as special cases.

\subsubsection{Decomposing the Clifford Hierarchy}
In this subsection, we provide two results regarding the synthesis of diagonal gates in the Clifford hierarchy. We focus on the subset of $\mathcal{CH}_l$ that are diagonal gates because they are more structured. It is known that this diagonal subset, which we denote as $\mathcal{D}_l$, forms a group \cite{cui2017diagonal} (to reiterate, $\mathcal{CH}_l$ is not a group for $l\ge 3$ and provides a complete gate set for universal quantum computation). $\mathcal{D}_l$ is generated by a set of controlled-phase gates.
\begin{proposition}
Denote the diagonal subset of $\mathcal{CH}_l$ as $\mathcal{D}_l$. $\mathcal{D}_l$ forms a group and is generated by $C^k Z^{2^{k-l+1}}$, $k=0 \ldots l-1$. $C^k Z^{2^{k-l+1}}$ acts on $k+1$ qubits and is defined as
\begin{equation}
    C^k Z^{2^{k-l+1}} = \sum_{x=0}^{2^{k+1}-2} \ketbra{x}{x} + e^{i \pi{2^{1+k-l}}} \ketbra{2^{k+1}-1}{2^{k+1}-1}
\end{equation}
Where $x$ denotes the literal value of the $k$-bit bitstring.
\end{proposition}

As an example, $\mathcal{D}_3$ is generated by $T$, $CS$, and $CCZ$ gates. If we set $k=0$, the single-qubit phase gate $Z^{2^{1-l}}$ in $\mathcal{D}_l$ rotates $\ket{1}$ by a phase $e^{i \pi2^{1-l}} $ which is exponentially small in $l$. In other words, higher Clifford hierarchies contain rotations with smaller angles.

We now show that the single-qubit phase gate $Z^{2^{1-l}}$, together with the $CNOT$ gate, is already enough to generate any gate in $\mathcal{D}_l$. Moreover, $Z^{2^{1-l}}$ can be placed in one layer.
\begin{theorem}\label{ch_synthesis}
    $C^k Z^{2^{k-l+1}}$ can be synthesized from $CNOT$ gates and one layer of $Z^{2^{1-l}}$ gates or its inverse after appending ancilla qubits for all $k=0 \ldots l-1$.
\end{theorem}
\begin{proof}
    We apply a result from \cite{schuch2003programmable} which gives a procedure to synthesize an arbitrary phase gate using $CNOT$ gates and single-qubit rotations.
    \begin{lemma}\label{small_angle_synthesis}
        Given a diagonal phase gate $D$ acting on $k$ qubits such that $D \ket{\vec{x}} = e^{i \theta(\vec{x})}\ket{\vec{x}}$, where $\vec{x}$ is a $k$-bit bitstring labeling the computational basis. $D$ can be synthesized from $CNOT$ gates and $2^k$ single-qubit diagonal gates $R_{\vec{y}}$, where $\vec{y}$ is a $k$-bit bitstring. If we let $R_{\vec{y}}=\ketbra{0}{0}+e^{i\phi(\vec{y})}\ketbra{1}{1}$, then $\phi(\vec{y})$ is related to $\theta(\vec{x})$ by
        \begin{equation}\label{small_angle_fourier}
            \phi(\vec{y})  = \sum_{\vec{x}}\frac{1}{2^{k-1}} (-1)^{\vec{x}.\vec{y}} \theta(\vec{x})
        \end{equation}
    \end{lemma}
    We apply the above lemma to synthesize $C^k Z^{2^{k-l+1}}$ which acts on $k+1$ qubits. In this case we have $\theta(1^{k+1})= \frac{\pi}{2^{l-k-1}}$ and all other $\theta(\vec{x})=0$. Plugging these values into Eq. (\ref{small_angle_fourier}), we have $\phi(\vec{y}) = \pm \frac{ \pi}{2^{l-1}}$, $\forall \vec{y}$ which is exactly $Z^{2^{1-l}}$ or its inverse. This establishes that $C^k Z^{2^{k-l+1}}$ can be synthesized with $CNOT$ gates,  $Z^{2^{1-l}}$ gates and its inverse. Finally, since $CNOT$ gates and  $Z^{2^{1-l}}$ gates are both almost classical gates, we apply Lemma \ref{almost_classical_one_layer} to put  $Z^{2^{1-l}}$ gates and its inverse into one layer, after appending ancilla.
\end{proof}
As a immediate corollary, the $CS$ gates in Fig. \ref{fig:cccz}(a) can be synthesized using one layer of $T^{\pm 1}$ gates, and the $CCCZ$ gate can be synthesized using one layer of $T^{\pm \frac{1}{2}}$ gates which is given in Fig. \ref{fig:cccz}(b).

Next, we discuss the synthesis of gates in $\mathcal{CH}_l$ using multiple layers of gates in $\mathcal{CH}_{m}$, where $m<l$. We show that a $C^{l}Z$ gate can be synthesized using two layers of gates in $\mathcal{CH}_{m}$ when $l$ is not too big.
\begin{theorem}
     The $C^{l}Z$ gate, where $l \le 2m$, can be synthesized from Clifford gates and two layers of diagonal gates from $\mathcal{D}_{m+1}$, after appending one clean ancilla.
\end{theorem}
\begin{proof}
    We give an explicit construction in Fig. \ref{fig:cnz} that generalizes the construction in \cite{gidney2021cccz}. We first explain the resource requirement in the second equality. The construction consists of two $C^{m+1} Z$ gates, two $C^{m'+1}Z $ gates, a $C^m S$ gate, a $C^{m'}S$ gate, and a Clifford gate $X^{-\frac{1}{2}}=H S^{-1} H$. One can see that $C^{m+1} Z$ and $C^m S$ gates belong to $\mathcal{D}_{m+1}$ while $C^{m'+1} Z$ and $C^{m'} S$ gates belong to $\mathcal{D}_{m'+1}$. By setting $m=m'$, we can synthesize the $C^lZ$ gate where $l=2m$ which gives the upper bound on $l$.

    Next, we explain the correctness of this construction. The construction is based on the following identity:
    \begin{equation}
        (-1)^{ab} = i^a i^b (-i)^{a \oplus b}
    \end{equation}
    where $a, b \in \{0,1\}$ are Boolean variables, $ab$ denotes $a$ AND $b$, and $a \oplus b$ denotes a XOR $b$. Now we set $a$ to be the AND product of $m$ bits $a= a_1a_2 \ldots a_m$, set $b$ to be the AND product of $m'$ bits $b= b_1b_2 \ldots b_{m'}$. $(-1)^{ab}$ is exactly a $C^lZ$ gate acting on $a_1, a_2, \ldots, a_m, b_1, b_2, \ldots, b_{m'}$. 
    
  The first equality in Fig. \ref{fig:cnz} reflects the right-hand side of the above identity. We introduce an ancilla qubit and use $m$-Toffoli and $m'$-Toffoli gates to compute $a \oplus b$. We then apply a $S^{-1}$ gate to introduce a phase $(-i)^{a \oplus b}$ and apply a $C^m S$ gate and a $C^{m'}S$ gate to introduce a phase $i^a i^b$. Lastly, we apply the Toffoli gates again to uncompute the ancilla. Writing the Toffoli gates as $C^m Z$ gates sandwiched by Hadamard gates on the target bit, we obtain the second equality.
\end{proof}
\begin{figure}
\includegraphics[width=0.9\linewidth]{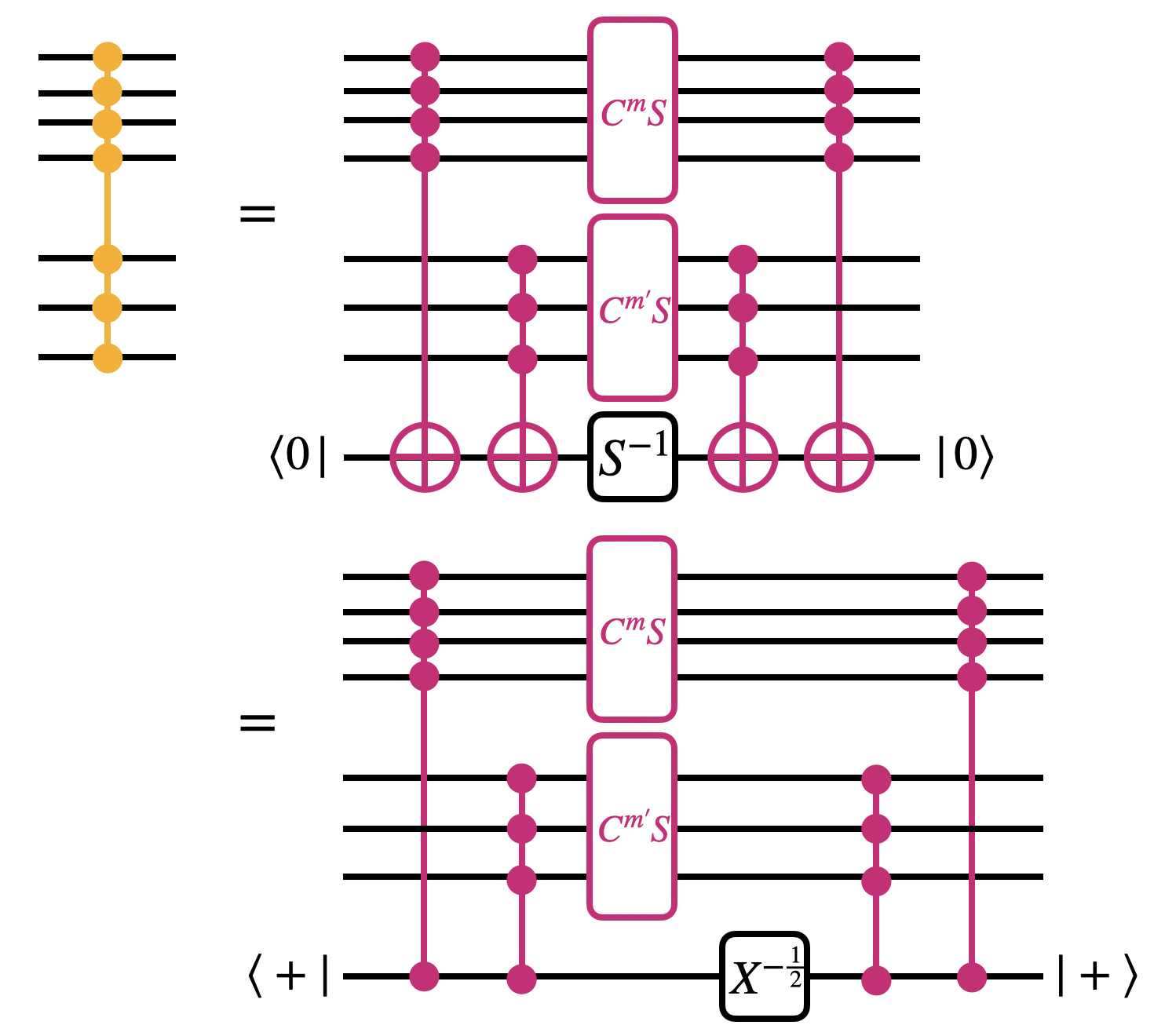}
\caption{\label{fig:cnz}\textbf{Synthesis of the $C^l Z$ gate} Synthesizing $C^lZ$ gates with two layers of diagonal gates from $\mathcal{D}_{m+1}$, where $l \le 2m$. The magenta gates are in $\mathcal{D}_{m+1}$, while the gate $X^{-\frac{1}{2}}$ is a Clifford gate. $X^{-\frac{1}{2}}$ prevents the two magic layers from being parallelized into one layer.
}
\end{figure}
Lastly, the above construction does not allow the synthesis of generic gates in $\mathcal{D}_{l}$ using two layers of gates in the lower Clifford hierarchy. For example, to our best knowledge, currently the best construction to synthesize a $CT$ gate, which is in $\mathcal{D}_4$, takes three layers of gates in $\mathcal{D}_3$ \cite{ct_gate}.

\begin{figure*}
\includegraphics[width=0.9\linewidth]{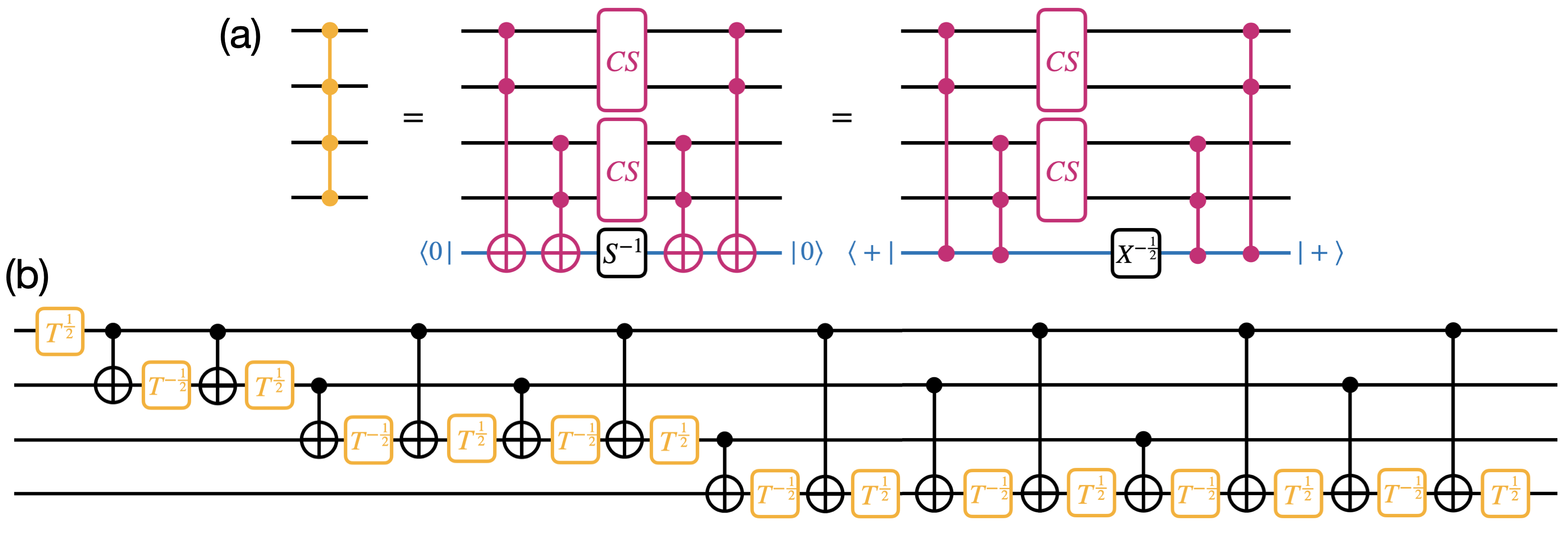}
\caption{\label{fig:cccz}\textbf{Synthesis of the $CCCZ$ gate.} (a) Synthesizing $CCCZ$ gate with two layers of $CCZ$ and $CS$ gates, equivalently two layer of $T^{\pm 1}$ gates after including ancilla. (b) Synthesizing $CCCZ$ gate with $CNOT$ gates and $T^{\pm \frac{1}{2}}$ gates. The circuit can be compiled to have one layer of $T^{\pm \frac{1}{2}}$ gates.}
\end{figure*}

\subsubsection{Proof of Hardness}
With the ingredient of the Hadamard test and the decomposition of the $CCCZ$ gate, we are ready to establish the hardness of computing Pauli observable in $T$-depth-two and $T^{\frac{1}{2}}$-depth-one circuits.
\begin{theorem}
       Given a circuit with one layer of non-Clifford gate $U= (\prod_i T_i^{\frac{k_i}{2}}) U_{c,l}$, where $T_i^{\frac{1}{2}}$ acts on the qubit $i$ and $k_i$ denotes some integer power, then computing Pauli observable $\braket{P}=\bra{0} U^\dag P U \ket{0}$ up to a $1$ multiplicative error is GapP-complete
\end{theorem}
\begin{proof}
    We first construct a Hadamard test circuit (Fig. \ref{fig:3ch_4ch}(c)) to reduce the task of computing the amplitude of any degree-three IQP circuit to the task of computing the Pauli observable in a circuit with $CZ$, $CCZ$, and $CCCZ$ gates. Next, we use the parallelization trick (Fig. \ref{fig:parallelization}) to parallelize the diagonal gates. Then, we compile $CCZ$ gates to one layer of $T^{\pm 1}$ gates (Fig. \ref{fig:ccz}(a)) and compile $CCCZ$ gates to one layer of $T^{\pm \frac{1}{2}}$ gates (Lemma \ref{almost_classical_one_layer} and Fig. \ref{fig:ccz}(b)). Naturally, $T^{\pm 1}$ gates are integer powers of $T^{\frac{1}{2}}$. Therefore, the hardness of computing Pauli observable in $T^{\frac{1}{2}}$-depth-one circuits follows from the hardness of computing the amplitude of the degree-three IQP circuit.
\end{proof}

\begin{theorem}
       Given a circuit with two layers of T gates $U=\prod_i (\prod_i T_i^{k_{m,i}}) U_{c,m} \prod_i (\prod_i T_i^{k_{l,i}}) U_{c,l}$, where $T_i$ acts on the qubit $i$ and $k_{l,i}$, $k_{m,i}$ denote some integer power, then computing the Pauli observable $\braket{P}=\bra{0} U^\dag P U \ket{0}$ up to a $1$ multiplicative error is GapP-complete
\end{theorem}
\begin{proof}
    The proof is similar to the proof of the previous theorem, except we decompose the $CCCZ$ gate into two layers of $CCZ$ and $CS$ gates, shown in Fig. \ref{fig:cccz}(a). Then, using Lemma \ref{ch_synthesis}, both $CCZ$ and $CS$ gates can be compiled into one layer of $T^{\pm 1}$ gates (Lemma \ref{almost_classical_one_layer}). The entire circuit then contains two layers of $T^{\pm 1}$ gates. 
\end{proof}

\section{Estimating Observable in Magic-Depth-One Circuits}

We have seen the easiness and hardness of computing the amplitude and the Pauli observable, up to a small multiplicative error, in magic-depth-one circuit. Nevertheless, a quantum computer computes amplitudes or Pauli observable only up to 1/poly($n$) additive error in polynomial time because it is a sampling machine. Recall that a $\epsilon$ multiplicative error means that the estimate $\tilde{z}$ deviates from the ground truth $z$ by $|\tilde{z}-z| \le \epsilon |z|$, while $|z|$ can be exponentially small in $n$. On the other hand, an $\epsilon$ additive error only requires that $|\tilde{z}-z| \le \epsilon$. One can see that having a small additive error is a more relaxing constraint than having a small multiplicative error when $|z|$ is small.

While we have shown that sampling from magic-depth-one circuit is classically hard under plausible complexity-theoretic conjectures, we give a polynomial-time classical algorithm to compute both amplitudes and Pauli observable up to 1/poly($n$) additive error. 
%

\subsection{Estimating Observable via Sampling an Auxiliary Distribution}
The main idea of the classical algorithm is to find an auxiliary sampling problem that produces the same Pauli observable values. Since estimating the Pauli observables becomes a sampling problem, one can also estimate any observable that is a uniform superposition of $A$ Pauli observables $P=\frac{1}{A} \sum_a P_a$. One simply has to sample $P_a$ from the uniform distribution and then estimate $P_a$ via sampling. This gives an estimate of $P$ up to a small additive error.

\vtwo{
\begin{theorem}
\label{thm:additive}  Given a circuit with one layer of diagonal non-Clifford gates $U=D U_{c,l}$, where $D$ is a diagonal gate. whose elements can be computed in time $t(n)$. Suppose we want to estimate an observable that can be written as the uniform average of $A$ Pauli observable: $P=\frac{1}{A} \sum_{a=1}^{A} P_a$. Then there exists a classical algorithm to estimate $\bra{0} U^\dag P U \ket{0}$ up to $\epsilon$ additive error and with $1-\delta$ probability in time $O(n^3 + \max(n, t(n))\frac{ \log(2/\delta )}{ \epsilon^2})$.
\end{theorem}
}
One can see that the task of estimating Pauli observables is the case of $A=1$. Estimating probability in the computational basis corresponds to setting $P_a$ to be the full stabilizer group, in which case $A=2^n$. More generally, one can also estimate the marginal probability on $k$ qubits. For example, to estimate the probability that the first $k$ qubits are zero, we estimate the following observable:
\begin{equation}
    P = U_{c,r}^\dag\prod_{i=1}^{k} \frac{1+ Z_i}{2} U_{c,r}
\end{equation}
One can again see that this is the average of $A=2^k$ Pauli operators. \vtwo{We note that $D$ can be a non-local gate in general. We only demand the ability to compute the elements efficiently. Therefore, Theorem \ref{thm:additive} applies to $T$-depth-one circuits, $T^{\frac{1}{2}}$-depth-one circuits, and more generally, circuits with one layer of non-local diagonal gates. When $D$ is a product of local diagonal gates, the time complexity of computing elements $t(n)=O(n)$.}

\begin{figure*}
\includegraphics[width=\linewidth]{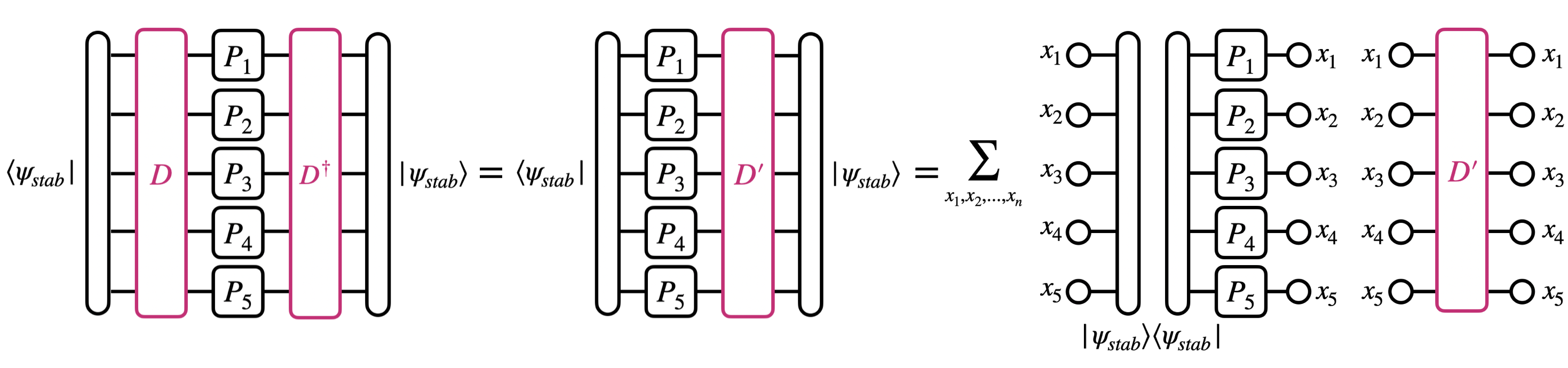}
\caption{\label{fig:diagonal_sampling}\textbf{Constructing the auxiliary sampling problem.} Converting the problem of computing Pauli observable to sampling from the diagonal distribution of $\ketbra{\psi_{stab}}{\psi_{stab}}P$ (if it is non-trivial) and evaluating the expectation value of some other observable $\prod_i D_i'$, defined in Lemma \ref{diagonal_push}.}
\end{figure*}

\vtwo{
We will use the following lemma in constructing the classical algorithm.
\begin{lemma}\label{diagonal_push}
   For any Pauli operator $P$ and any Diagonal gate $D$,
    \begin{equation}
        D^\dag P D = P D'
    \end{equation}
    Where $D'$ is another diagonal unitary determined by $D$ and $P$. Moreover, by having query access to $D$, any on-diagonal element of $D'$ can be computed in $O(n)$ time.
\end{lemma}
\begin{proof}
We will construct $D'$ explicitly. First, notice that any Pauli operator $P$ is also an almost classical gate. In other words,
\begin{equation}
    P = \sum_x  e^{i \phi(x)}\ketbra{f(x)}{x}
\end{equation}
We write $D$ in the computational basis
\begin{equation}
    D = \sum_x e^{i \phi_{D}(x)} \ketbra{x}{x}
\end{equation}
Now we expand $D^\dag P D$ in the computational basis
\begin{equation}
    D^\dag P D = \sum_x e^{-i [\phi_{D}(f(x))-\phi_{D}(x)]} e^{i \phi(x)}\ketbra{f(x)}{x}
\end{equation}
$D'$ can be defined as
\begin{equation}\label{push_thru_def}
    D'=\sum_x e^{-i [\phi_{D}(f(x))-\phi_{D}(x)]} \ketbra{x}{x}
\end{equation}
Lastly, given query access to $\phi_{D}(x)$ and given that $f(x)$ is computable in $O(n)$ time, any on-diagonal elements can be computed in $O(n)$ time.
\end{proof}
The above lemma ``pushes'' the diagonal operator to the right-hand side. Moreover, computing any elements in $D'$ is computationally efficient, given query access to elements of $D$.
}

\begin{proof}[Proof of Theorem \ref{thm:additive}]
We will explicitly construct the classical algorithm. For simplicity, we begin by setting $A=1$, in other words, we estimate a Pauli operator. The essence of this algorithm is to generate a sampling problem that allows us to estimate Pauli observables. We consider the stabilizer state $\ket{\psi_{stab}}=U_{c,l}\ket{0}$. The Pauli observable $\langle P \rangle$ can be be expressed as 
\begin{equation}
    \langle P \rangle = \Tr[\ketbra{\psi_{stab}}{\psi_{stab}}  D^\dag P D]
\end{equation}
To proceed, we exploit the diagonal structure of $D$ and apply Lemma \ref{diagonal_push} to push $D$ to the right-hand side. After that, we express the expectation value of $P$ by summing over the computational basis.
\begin{align}
    \langle P \rangle &= \Tr[\ketbra{\psi_{stab}}{\psi_{stab}}P D'] \\
    &= \sum_x \braket{x|\psi_{stab}}\bra{\psi_{stab}}P \ket{x}  e^{i \phi_{D'}(x)}
\end{align}
The above equation is visualized in Fig. \ref{fig:diagonal_sampling}. In the first equality, we apply Lemma \ref{diagonal_push}. In the second equality, we calculate the trace by summing over the computational basis. $e^{i \phi_{D'}(x)}=\bra{x} D'\ket{x}$ denotes the phase $ D'$ applies to $\ket{x}$. 

To proceed, we show how to find non-trivial on-diagonal elements $\braket{x|\psi_{stab}}\bra{\psi_{stab}}P \ket{x}$ via canonicalizing the stabilizer tableau \cite{garcia2012efficient}. The canonicalized stabilizer tableau is shown in Fig. \ref{fig:canon_solve}(a). Each row corresponds to a stabilizer generator. In the canonicalized form, the stabilizer tableau contains a $X$ diagonal and a $Z$ diagonal. Along the $X$ diagonal, the elements are all $X$ or $Y$, while along the $Z$ diagonal, the elements are all $Z$. Elements above the $X$ diagonal (orange) contain arbitrary Pauli operators. Elements between the $X$ and $Z$ diagonals contain only $I$ or $Z$. Elements below the $Z$ diagonal contain only $I$.

The two diagonals distinguish two types of stabilizer generators. We denote generators on the $X$ diagonal as $S_{X,j}$ and denote generators on the $Z$ diagonal as $S_{Z,k}$, where $j$ and $k$ are labels.  $S_{X,j}$ contains at least one $X$ or $Y$ on the $X$ diagonal, whereas $S_{Z,k}$ contains no $X$ or $Y$ at all.

With the canonicalized stabilizer generators, we evaluate the on-diagonal elements $\braket{x|\psi_{stab}}\bra{\psi_{stab}}P \ket{x}$ by expanding it into a product of stabilizers.
\begin{equation}
    \braket{x|\psi_{stab}}\bra{\psi_{stab}}P \ket{x} = \frac{1}{2^n} \bra{x} \prod_k (I+S_{Z,k}) \prod_j (I+S_{X,j}) P \ket{x}
\end{equation}
The expression now contains a sum of exponentially many Pauli expectations, but many terms are zero. To see that, note that if a Pauli string contains $X$ or $Y$, then its on-diagonal matrix elements are all zeros. In the above expression, $X$ and $Y$ originate from $\prod_j (I+S_{X,j}) P$. Therefore, we would like to find terms in $\prod_j (I+S_{X,j}) P$ with no $X$ or $Y$.

It turns out that if such a term exists, then it is \emph{unique}, which we denote as $\tilde{P}$. To find $\tilde{P}$ and show its uniqueness, we exploit the diagonal structure of the $S_{X,j}$ part of the tableau. First, finding $\tilde{P}$ can be thought of as using a subset of $S_{X,j}$ to cancel out $X$ and $Y$ in $P$. Specifically, we decompose $P$ into $P=P_X P_Z$, where $P_X$ contains only $I$ and $X$, and $P_Z$ only contains $I$ and $Z$. We set $P_X$ to have the $+1$ sign and absorb any possible phases in $P_Z$. Similarly, we decompose all $S_{X,j}$ into  $S_{X,j}= S_{X,j}^{(Z)}S_{X,j}'$, where $S_{X,j}'$ contains only $I$ and $X$, and $S_{X,j}^{(Z)}$ contains $I$ and $Z$. Again we set $S_{X,j}'$ to have the $+1$ sign and absorb any possible phases in $S_{X,j}^{(Z)}$. After the decomposition, we have
\begin{equation}
    \prod_j (I+S_{X,j}) P = \prod_j (I+S_{X,j}^{(Z)}S_{X,j}') P_X P_Z
\end{equation}
One can think about the above procedure as ``ignoring'' the $Z$ component from $S_{X,j}$ and $P$. To remove $X$ and $Y$, we would need to find a subset of $S_{X,j}'$ that cancels out $P_X$. Specifically, we define a bitstring $\vec{s} \in \mathbb{F}_2^{|S_{X,j}|}$ with length equal to the number of $X$ type stabilizers, denoted as $|S_{X,j}|$. Canceling out $P_X$ is equivalent to solving the following linear equation:
\begin{equation}
    \prod_j S_{X,j}' s_j = P_X
\end{equation}
The above equation is depicted in Fig. \ref{fig:canon_solve}(b). The first term $S_{X,j}'$ corresponds to the upper half of the stabilizer tableau, with $Y$ replaced with $X$ and $Z$ replaced with $I$. Crucially, the above equation can be thought of as an under-determined equation over a finite field $\mathbb{F}_2$. Therefore, the solution does not have to exist. If it does not exist, $\langle P \rangle=0$. On the other hand, if the solution exists, it has to be unique because the tableau of $S_{X,j}'$ is already in an upper-triangular form, and one can find the solution by performing the standard substitution.

Suppose the solution exists, then after canceling out $P_X$ with $\prod_j S_{X,j}' s_j$, The remaining $Z$ component, in other words $\tilde{P}$, consists of
\begin{equation}
    \tilde{P} = (\prod_j S_{X,j}^{(Z)} s_j) P_Z
\end{equation}
With $\tilde{P}$, we can finally write the expression of $\langle P \rangle$ in the following diagonal form.
\begin{align}
    \langle P \rangle &= \frac{1}{2^n} \sum_x  \bra{x} \prod_k (I+S_{Z,k}) \tilde{P} \ket{x} \prod_i e^{i \phi_{D'_i}(x)} \\
    &= \frac{1}{2^n} \sum_x  \bra{x} \prod_k (I+S_{Z,k}) \ket{x}e^{i \phi_{\tilde{P}}(x)} \prod_i e^{i \phi_{D'_i}(x)}
\end{align}
Where in the second line, we use the fact that $\tilde{P}$ is diagonal to take it out of the bracket and replace it with $e^{i \phi_{\tilde{P}}(x)} = \bra{x} \tilde{P} \ket{x}$. The above equation can be considered as taking the expectation value of $e^{i \phi_{\tilde{P}}(x)} \prod_i e^{i \phi_{D'_i}(x)}$ over the diagonal distribution of a mixed stabilizer state generated by $S_{Z,k}$. Crucially, the diagonal distribution is a uniform distribution over the affine subspace, so it can be easily sampled. To be concrete, the mixed stabilizer state has the following diagonal form when written in the computational basis:
\begin{equation}\label{mixed_stab_comp_basis}
    \frac{1}{2^n} \prod_k (I+S_{Z,k}) = \frac{1}{2^r}\sum_{t=0}^{2^r-1} \ketbra{At + b}{At + b}
\end{equation}
Where $t: \mathbb{F}_2^r$ denotes a length-$r$ bitstring, $A: \mathbb{F}_2^{n \times r}$ and $b: \mathbb{F}_2^{n}$ can be derived from $S_{Z,k}$ in $O(n^3)$ time \cite{dehaene2003clifford}. Therefore, one can sample $t$ from the uniform distribution and estimate $\langle P \rangle$ accordingly.
\begin{figure}
\includegraphics[width=0.9\linewidth]{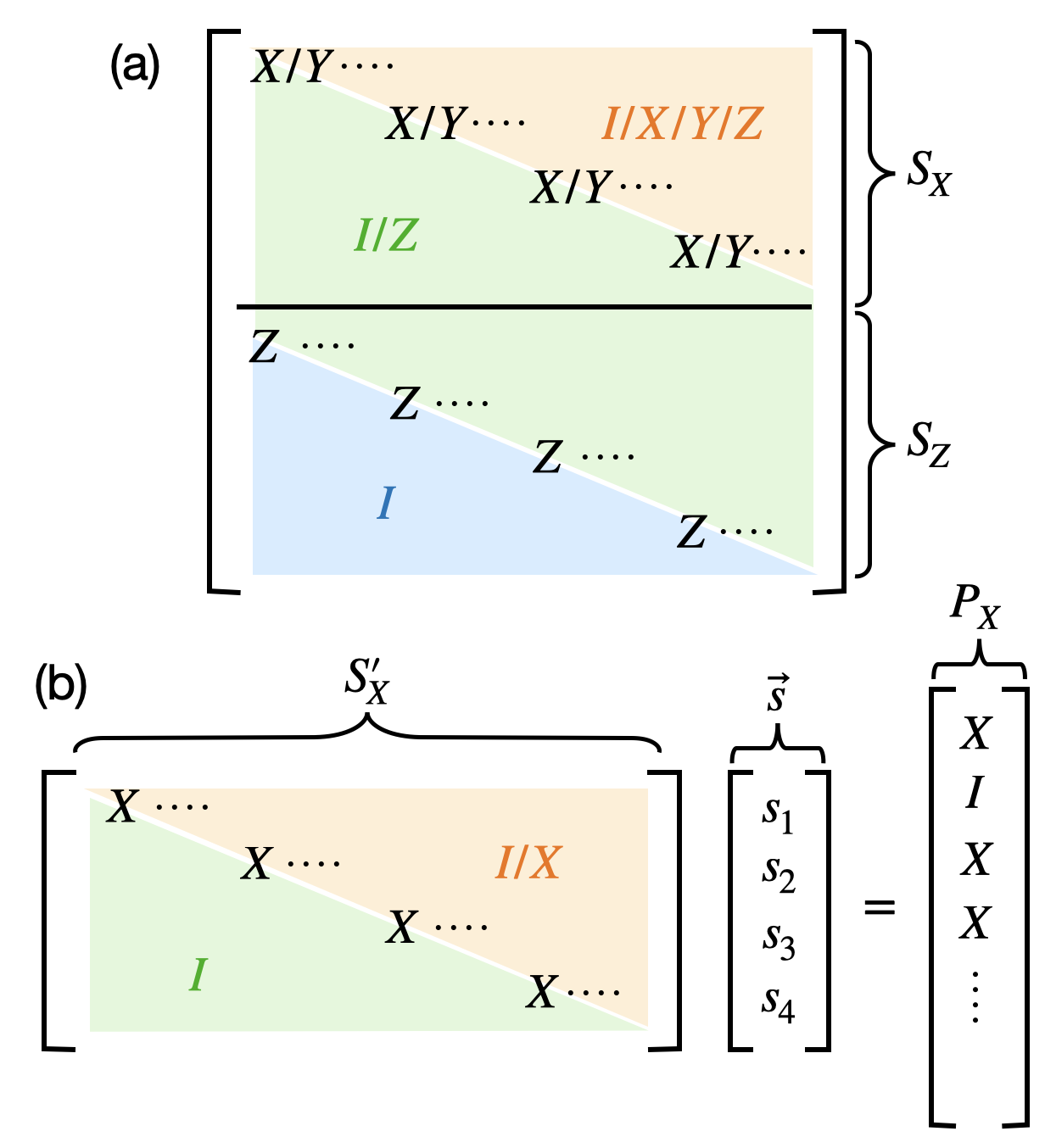}
\caption{\label{fig:canon_solve}(a) The canonicalized stabilizer tableau (b) The linear equation to cancel out $P_X$. The tableau of $S_X'$ is identical to the $S_X$ part of (a), but replacing $Y$ with $X$ and replacing $Z$ with $I$, and changing the sign to $+1$.}
\end{figure}

Finally, in the case where $P$ is the uniform average of many Pauli operators ($A>1$), one can sample $P_a$ from the uniform distribution and estimate the expectation via the above procedure. Note that the mixed stabilizer state in Eq. (\ref{mixed_stab_comp_basis}) does not depend on the observable, so in practice, one samples $P_a$ and $t$ simultaneously in spirit of Monte Carlo sampling.

We detail the algorithm in Algorithm \ref{alg:additive}. To analyze the time complexity, first notice that the pre-processing step in line 1-3 has a time complexity of $O(n^3)$ due to the canonicalization as well as finding $A$, $b$, and $r$. Next, lines 7-8 take $O(n)$ time, and line 9 also takes $O(n)$ time because the stabilizer tableau of $S_{X,j}'$ (first term in Fig. \ref{fig:canon_solve}(b)) is already in the upper-triangular form so one only needs to perform substitutions row by row. \vtwo{Line 13 takes $O(\max(n,t(n)))$ time.} Finally, the standard Chernoff bound gives the sample complexity $M=\log(\frac{2}{\delta})/\epsilon^2$, which gives the total time complexity of $O(n^3 +  \max(n,t(n))\frac{\log(2/\delta )}{ \epsilon^2})$. When $A=1$, lines 8-11 only need to be performed once.
\end{proof}

\begin{algorithm}[H]
\textbf{Input:}  $U_{c,l}$, $D$, $\{P_a\}$, $\epsilon$, $\delta$ \\
\textbf{Output:} estimate of $\bra{0} U_{c,l}^\dag D^\dag(\frac{1}{A}\sum_a P_a) D  U_{c,l} \ket{0}$
\caption{Evaluating observable of magic-depth-one circuit up to $\epsilon$ error with probability $1-\delta$}\label{alg:additive}
\begin{algorithmic}[1]
\State compute $\{S_{X,j}\}$, $\{S_{Z,k}\}$ by canonicalizing $U_{c,l}\ket{0}$
\State decompose $S_{X,j}=S_{X,j}^{(Z)}S_{X,j}', \forall i$
\State find $A$, $b$, $r$ from $\{S_{Z,k}\}$ (Eq. (\ref{mixed_stab_comp_basis}))
\State $M=log(\frac{2}{\delta})/\epsilon^2$
\State define an array of samples $\vec{P_m}$, $m=1...M$
\For{$m$ in $\{1,2,...,M\}$}
\State sample $P_a$ uniformly, sample $t$ uniformly
\State decompose $P_a=P_{X,a}P_{Z,a}$
\State solve $\prod_j S_{X,j}' s_j = P_X$ for $\vec{s}$
\If {solution exists}
    \State $\tilde{P} = (\prod_j S_{X,j}^{(Z)} s_j) P_Z$
    \State $e^{i \phi_{\tilde{P}}(At+b)} = \bra{At+b} \tilde{P} \ket{At+b}$
    \vtwo{\State compute $\phi_{D'}(x)$ from $D$, $P_a$, and $x$ (Eq. (\ref{push_thru_def}))}
    \State $\vec{P}_m = e^{i (\phi_{\tilde{P}}(x)+\phi_{D'}(x))}$
\Else
    \State $\vec{P}_m =0$
\EndIf
\EndFor
 \State \Return  median-of-mean estimate of $P$ from $\vec{P}_m$
\end{algorithmic}
\end{algorithm}

\subsection{Hardness of Sampling the Original Distribution}
One may wonder if estimating the Pauli observables, amplitudes, and more generally marginal probability distributions enables sampling from the distribution approximately. If this happens, then either the average-case hardness conjecture in \cite{bremner2016average} is false or the polynomial hierarchy collapse to the third level. Both of them seem unlikely to happen.

We give strong evidence that estimating marginal probability distributions does not allow for the approximate sampling from the entire distribution. With access to the marginal probability distributions, the typical strategy to sample the entire distribution is via the bit-by-bit sampling: one first sample the first bit from its marginal distribution, then sample the second bit conditioned on the first bit being the sampled value, then sample the third bit conditioned on the first two bits, and so on. This would require the computation of the probability of, say the $k$-th bit conditioned on bit $1 \ldots k-1$. The conditional probability is related to the marginal probability by
\begin{equation}
    P(x_k | x_{k-1} \ldots x_{1}) = \frac{P(x_k x_{k-1} \ldots x_{1})}{P(x_{k-1} \ldots x_{1})}
\end{equation}
Where $x_i \in \{0,1\}$ denotes the $i$-th bit. crucially, one can only estimate the denominator $P(x_{k-1} \ldots x_{1})$ up to polynomially small additive error. When $k$ becomes $O(n)$, the true value of $P(x_{k-1} \ldots x_{1})$ typically becomes exponentially small, so the error is way bigger than the ground truth. Therefore, the error in the denominator results in a big error in the conditional probability. The above analysis strongly suggests that estimating observable up to a small additive error is strictly weaker than sampling from the distribution up to a small total variation distance.

\subsection{Classical Algorithm for Sampling the Marginal Distribution}
While we have seen strong evidence that estimating observables up to a small additive error is insufficient to sample from the original distribution, we show that it is possible to sample from the marginal distribution of a sufficiently small subsystem with $k$ qubits. We accomplish this by computing the $2^k$ marginal probabilities to sufficient accuracy using Algorithm \ref{alg:additive}
.

\begin{corollary}
Given a circuit with one layer of diagonal non-Clifford gates $U=U_{c,r}(\prod_i D_i) U_{c,l}$, there exists a classical algorithm to sample from the marginal distributions of $k$ qubits $\epsilon$ close to the actual distribution in the total variation distance (Eq. (\ref{tvd})) and with $1-\delta$ probability in time $O(2^k(n^3+4^{k-1} n\epsilon^{-2}\log(2^{k+1}/\delta)))$
\end{corollary}
\begin{proof}
We sample by computing the $2^k$ marginal probabilities using Algorithm \ref{alg:additive} up to error $\epsilon'$ and with probability $1-\delta'$. Using the union bound, the total failure rate $\delta$ is given by $\delta = O(2^k \delta')$. The total variation distance $\epsilon$ is upper-bounded by $\epsilon \le \frac{1}{2} 2^k \epsilon'$. Plugging the relation into Theorem \ref{thm:additive} to get the time complexity $O(4^{k-1}n \epsilon^{-2} \log(\frac{2^{k+1}}{\delta}))$ for estimating each marginal probability. Finally, one has to repeat for all $2^k$ marginal probabilities which gives the stated time complexity.
\end{proof}
The above corollary suggests that sampling from the marginal of $k=O(\log(n))$ qubits up to a 1/poly($n$) error is classically efficient.

\section{Path Integral}
Lastly, we discuss the classical simulation of quantum circuits beyond one magic layer. While we do not have a polynomial-time classical algorithm here, nor do we expect one, the shallow magic depth can still be exploited to reduce the cost of classical simulations. We accomplish this by performing a path integral at each magic layer.

Suppose we want to compute the amplitude of a unitary $U$ that contains $d$ layers of diagonal magic gate.
\begin{equation}\label{d_layer_magic}
    U = D_{d} U_{c,d} D_{d-1}\ldots U_{c,2} D_{1} U_{c,1}
\end{equation}
Where $U_{c,i}$ denotes some Clifford unitary and $D_i$ denotes the diagonal magic gate. We do not require $D_i$ to factorize into products of local gates, but merely require that each entry $\bra{x}D_i\ket{x}$ can be computed efficiently. We show that there exists a path integral algorithm that scales favorably in the magic depth $d$ then other methods.
\begin{theorem}
Given the unitary $U$ defined in Eq. (\ref{d_layer_magic}), there exists a classical path integral algorithm to compute $\bra{x}U\ket{0}$ in time $O((t(n)+n^3) (2d)^{n+1})$, where $t(n)$ denotes the runtime to compute $\bra{x}D_i\ket{x}$, and with space $O(n^2 + n\log(n))$.
\end{theorem}
\begin{proof}
 To begin with, we show that computing $\bra{x} D_i U_{c,i} \ket{y}$, where $\ket{x}$ and $\ket{y}$ are computational basis states, can be computed in time $O(t(n)+n^3)$. To see that, first realize that $U_{c,i} \ket{y}$ can be written in the computational basis in time $O(n^3)$  (Eq. (\ref{clifford_comp_basis})). Next, $\bra{x}D_i = \bra{x}D_i\ket{x} \bra{x}$ and computing $\bra{x}D_i\ket{x}$ takes time $t(n)$. Finally, $\bra{x} U_{c,i} \ket{y}$ simply retrieves the term from Eq. (\ref{clifford_comp_basis}). The space complexity is $O(n^2)$ which comes from storing the affine subspace and the quadratic form in Eq. (\ref{clifford_comp_basis}). Notice that unlike the hardness of computing the amplitude in Theorem \ref{thm:hardness_amplitude}, computing the amplitude here is classically easy because there is only one Clifford unitary on one side.
 
 Next, we follow \cite{aaronson2016complexity} and perform the path integral recursively. in the base case $d=1$, we compute $\bra{x}U\ket{0}$ directly in time $O(t(n)+n^3)$. For generic $d$, we insert an identity $I=\sum_y \ketbra{y}{y}$ at layer $\left \lfloor{\frac{d}{2}}\right \rfloor$.
 \begin{equation}
 \begin{split}
     \bra{x}U\ket{0} = \sum_y \bra{x}D_{d} U_{c,d} D_{d-1}\ldots U_{c,\left \lfloor{\frac{d}{2}+1}\right \rfloor}  \ket{y} \\
    \times  \bra{y} D_{\left \lfloor{\frac{d}{2}}\right \rfloor} \ldots U_{c,2} D_{1} U_{c,1} \ket{0}
      \end{split}
 \end{equation}
Therefore, we reduce the problem to computing $2^{n+1}$ amplitudes with $\left \lfloor{\frac{d}{2}}\right \rfloor$ layers of magic gate and summing them up. Applying the above process recursively until $d=1$. There are at most $2^{\left \lceil{\log(d)}\right \rceil (n+1)} \le (2d)^{n+1}$ amplitudes to compute, thus the runtime is $O((t(n)+n^3)(2d)^{n+1})$. Storing the bitstrings from the recursion takes space $O(n \log(d))$, and there is a space cost of $O(n^2)$ in computing amplitude but it does not carry over the recursion, leaving a space complexity of  $O(n^2 + n\log(n))$.
\end{proof}
When $D_i$ factorizes into products of local gates, $t(n)=O(n)$, and thus the time complexity becomes $O(n^3(2d)^{n+1})$. Crucially, the scaling with magic depth is sub-exponential. This should be compared with low-stabilizer-rank simulations, where the time complexity is exponential in the total number of magic gates which is $O(dn)$. On the other hand, if one performs the standard path-integral simulations, the time complexity also depends on the number of Clifford gates. Finally, the state vector simulation has favorable scaling in the number of magic gates but has an exponential memory cost. Therefore, in the regime where there are many Clifford gates, yet an extensive number of magic gates concentrates over a few layers, our algorithm provides a significant speedup over other methods. 

\section{Discussion}
We have systematically investigated the classical simulability of quantum circuits with an extensive number of magic gates concentrating at one layer. The complexity depends on the type of task and the desired precision. We show that computing the amplitude in $T$-depth-one circuits is GapP-complete, while computing Pauli observable is in P. However, adding one more layer of $T$ gates or replacing $T$ with $T^\frac{1}{2}$ immediately increases the hardness of computing the Pauli observable to be GapP-complete. 

The above results hold up to a small multiplicative error. If one only demands 1/poly($n$) additive error, then estimating both amplitudes and Pauli observable can be performed classically in polynomial time, while one can sample from any $\log(n)$ sized marginal distributions. Sampling from the entire distribution is still classically hard, under certain plausible complexity conjectures. Lastly, we give a path integral algorithm that, despite scaling exponentially in $n$, scales favorably in the number of magic layers. We expect this algorithm to outperform other algorithms in the regime where extensive magic gates concentrate at a few layers.

Overall, our work provides new insights into the complexity of magical circuits, highlighting the importance of magic depth and the type of computational tasks that could drastically affect the hardness of simulations. In practice, we give a classical algorithm to estimate amplitudes, Pauli observable, and sample from a small marginal in magic-depth-one circuits to the same precision that BQP can achieve. This rules out the possibility of quantum advantages in magic-depth-one circuits by estimating amplitude or Pauli observable. One would need at least two layers of magic gates or take advantage of the full sampling power to achieve quantum advantages.

\subsection{Comparison with Existing Work}
We compare our results to the existing results in the literature. First, we show that our result does not challenge the hardness of BQP, in other words, the full power of quantum computing. It is known that in practice, putting all $T$ gates in the first layer of the circuit is already sufficient for universal quantum computation because one can perform magic state injection \cite{knill2004fault}. Does our easiness result of computing Pauli observable in $T$-depth-one circuits (Theorem \ref{thm:3ch_pauli}) imply that computing Pauli observable of generic quantum circuits is classically easy? 

This is not true because of the following: in magic state injection, one has to post-select the ancilla to be $\ket{0}$ or perform a feedback operation if the ancilla is measured to be $\ket{1}$. The feedback operation is equivalent to a $CS$ gate that is non-Clifford. On the other hand, if one post-select $k$ ancilla, then evaluating the expectation of $P$ becomes $P \otimes \ketbra{0^k}{0^k}_A$, where $\ketbra{0^k}{0^k}_A$ acts on the $k$ ancilla qubits. When $k$ is $\omega(\log(n))$, $P \otimes \ketbra{0^k}{0^k}_A$ cannot be written as a polynomial sum of Pauli operators, and thus one cannot compute Pauli expectations efficiently when injecting $\omega(\log(n))$ $T$ gates.

Another way is to estimate $P \otimes \ketbra{0^k}{0^k}_A$ using Algorithm \ref{alg:additive}. The issue here is that for every ancilla included, the post-selected probability decreases by 1/2. Therefore, when $k$ is $\omega(\log(n))$, the expectation value of $P \otimes \ketbra{0^k}{0^k}_A$ is super-polynomially small, so Algorithm \ref{alg:additive} cannot estimate it in polynomial time. Therefore, our results do not allow us to compute generic Pauli expectation values of any quantum circuit beyond $\log(n)$ $T$ gates.

Similarly, placing all $T$ gates in the last layer of a constant-depth Clifford circuit is also sufficient for universal quantum computation because one can realize measurement-based quantum computation \cite{raussendorf2003measurement}. Nevertheless, one needs post-selection here again, so the classical simulation becomes intractable after $\omega(\log(n))$ $T$ gates. The above analysis in fact reveals the power of post-selection: while we have shown that $T$-depth-one circuit is strictly weaker than BQP unless P=BQP, augmenting it with post-selection promotes its power to postBQP which is equal to PP \cite{aaronson2005quantum}. This shows a sharp complexity separation by adding the power of post-selection which has been commonly observed in literature.

Next, we compare our go results (Theorem \ref{thm:3ch_pauli} and Algorithm \ref{alg:additive}.) with the earlier results that are based on low-stabilizer-rank approximations ~\cite{bravyi2016improved,bravyi2016trading,bravyi2019simulation,kocia2020improved,qassim2021improved}. These algorithms can accomplish strong simulations up to small multiplicative error, or perform weak simulations by sampling from a distribution that is close to the actual distribution in the total variation distance. Although Theorem \ref{thm:3ch_pauli} is strictly stronger than the previous results, Algorithm \ref{alg:additive} cannot be directly compared with the early methods because it only provides an additive estimate. Depending on the setup, one might favor one different algorithms. If one only needs to estimate observable, such as in the variation quantum eigensolver, to the precision that BQP can achieve, then Algorithm \ref{alg:additive} is more favorable. On the other hand, there are instances, such as computing the cross-entropy benchmark, where an exponentially high precision is required~\cite{arute2019quantum}. In this case, Algorithm \ref{alg:additive} would not be favorable and strong simulations up to a small multiplicative error would be required.

Finally, we point out that in the special case of IQP circuits, there already exist classical polynomial algorithms to compute the Pauli expectation of degree-three IQP circuits \cite{shepherd2009temporally}, as well as estimating the Pauli observable and amplitudes in generic IQP circuits \cite{yung2020anti} up to a small additive error. These algorithms can be considered a special case of our Theorem \ref{thm:3ch_pauli} and Algorithm \ref{alg:additive}. Sampling from a $\log(n)$ marginal of any IQP circuit can also be performed efficiently using the gate-by-gate sampling algorithm \cite{bravyi2022simulate}. Nevertheless, our results generalize to any magic-depth-one circuits.

\subsection{Exploiting Magic Depth in Other Tasks}
We point out some other recent work that exploits the structure of magic depth in other tasks. The first example is quantum state learning. In \cite{lai2022learning}, the authors proposed a tomography procedure to efficiently learn states generated by $O(\log(n))$  $T$ gates concentrated in one layer.  This algorithm is a ``proper'' learner in the sense that it outputs a Clifford + $T$ circuit whose output approximates the state being learned. On the other hand, while there are other results that can efficiently learn states generated by $O(\log(n))$  $T$ gates, possibly at different layers \cite{grewal2022low,grewal2023efficient,grewal2024improved,leone2024learning,hangleiter2024bell,oliviero2024unscrambling,leone2024learningefficient}, these algorithms are not proper learners because they only generate a low stabilizer-rank representation of the state, but not the Clifford + $T$ circuit that generates it. Proper learning of magic states beyond $T$-depth one remains an open problem. 

Circuits with shallow magic depth have also been investigated in the context of quantum dynamics and phase transitions \cite{niroula2023phase,bejan2024dynamical,turkeshi2024magic}. In Ref. \cite{niroula2023phase}, the authors consider one layer of small-angle rotations sandwiched by two Clifford encoder and decoder. They use this circuit to model the effect of coherent error on error correction. They show that there exists a phase where the rotation angle is small and the stabilizer syndrome measurements automatically removes the magic generated by the rotation. When the rotation angle is big, the circuit is in another phase where the stabilizer syndrome measurements cannot remove the magic.

\subsection{Future Directions}
We highlight several future directions. First, so far we have treated the Clifford unitary as a ``black box'' and have not exploited the locality structure within. Since we do not expect interference between causally connected magic gates, it is natural to ask whether the locality of the Clifford unitary can be exploited to reduce the cost of classical simulation, going beyond concentrating all magic gates in one layer. 

Second, since $\mathcal{FH}_2$ already contains hard problems such as factoring, it would be interesting to find a circuit with two layers of magic that can solve a hard problem. One subtlety here is that in $\mathcal{FH}_2$, two layers of Hadamard gates sandwich a layer of almost-classical gates that can encode any functions computable in polynomial time. On the other hand, two layers of magic sandwich a Clifford unitary, and a Clifford unitary is more restrictive than functions computable in polynomial time. This is because a Clifford unitary can be uniquely specified by its action on all the $X$ and $Z$ operators, so they have fewer degrees of freedom than functions computable in polynomial time. Although this does not prevent magic-depth-two circuits from performing hard tasks, it presents an additional challenge in designing such circuits.

Lastly, so far we have primarily considered IQP circuits as our paradigmatic model of magic-depth-one circuits. It would be interesting to find other examples of magic-depth-one circuits that can provide additional insights into the power of magic-depth-one circuits. For example, the ability to compute the amplitude in IQP circuits already allows one to compute the amplitudes of generic magic-depth-one circuits because of their GapP completeness. Similarly, is it possible to find a subclass of magic-depth-one circuits such that sampling from them encompasses the hardness of sampling from generic magic-depth-one circuits? We leave this question to future work. 

\vtwo{The importance of this work goes beyond the study of a toy circuit model, as recent efforts indicate that a very large set of circuits can be converted into magic-depth-one circuits. Recently, Ref. \cite{fux2024disentangling} provided numerical evidence showing that random Clifford+$T$ circuits can be recompiled to a magic depth of one. Subsequently, Ref. \cite{liu2024classical} provides a more analytic understanding and an algorithm that `transports' mid-circuit magic gates to the first layer. Following these observations, two open questions arise: What types of circuits does this algorithm apply to? Can this algorithm in Ref.\cite{liu2024classical} be generalized to putting magic gates in one layer in the middle of the circuit? }


We also point out the implication of our results for future quantum computing experiments. Ref. \cite{niroula2023phase} has analyzed the effect of coherent errors modeled by one layer of small-angle rotations. While they have only considered stabilizer codes generated by random Clifford circuits and small system size, our technique allows for analyzing the effect of non-Pauli noise on generic stabilizer codes and at a much larger system size. 

Magic-depth-one circuits could also be potentially useful for benchmarking beyond the Clifford regime. Our result shows that while estimating observables up to 1/poly($n$) additive error is classically easy, sampling from the full distribution remains hard in the worst case. Therefore, magic-depth-one circuits perform hard tasks but there are probes to partially verify the distribution. Such behavior sits between the Clifford randomized benchmarking in which the entire computation is classically easy, and random circuit sampling which is hard to spoof but also hard to verify. This renders the magic-depth-one circuits attractive for benchmarking future fault-tolerant quantum computers. 

\begin{acknowledgments}
(Y.Z.)$^2$ would like to thank Qi Ye,  Weiyuan Gong, David Gosset, and Sarang Gopalakrishnan for useful discussions. Yifan Zhang acknowledges support from NSF QuSEC-TAQS OSI 2326767. Yuxuan Zhang acknowledges support from the Natural Science and Engineering Research Council (NSERC) of Canada, and support from the Center for Quantum Materials and the Centre for Quantum Information and Quantum Control at the University of Toronto. 
\end{acknowledgments}

\begin{widetext}
\appendix

\end{widetext}


\bibliography{apssamp}

\end{document}